\documentclass[journal]{IEEEtran}

\ifCLASSINFOpdf
\else
\fi

\usepackage{multirow}
\usepackage{stfloats}
\usepackage{balance}
\usepackage{hyperref}

\usepackage{caption}
\usepackage[normalem]{ulem}

\usepackage{empheq}
\usepackage{mathrsfs}
\usepackage{amsfonts}
\usepackage{mathtools}
\usepackage{mathrsfs}
\usepackage{graphicx}          % Include this line if your
\usepackage{amsmath}

\DeclareMathOperator*{\argmin}{arg\,min}
\usepackage{enumerate,subcaption,graphicx}
\usepackage{algorithm,algpseudocode}%algorithmicx,algpseudocode}
\algnewcommand{\algorithmicgoto}{\textbf{go to}}%
\algnewcommand{\Goto}[1]{\algorithmicgoto~\ref{#1}}%
\algnewcommand{\LineComment}[1]{\Statex \(\triangleright\) #1}
\algnewcommand{\LineCommentN}[1]{\Statex \hspace{1cm}\(\triangleright\) #1}
\usepackage{multirow}
\usepackage{stfloats}
\usepackage[dvipsnames]{xcolor}
\usepackage{amsthm}
\newtheoremstyle{mystyle}%                % Name
  {}%                                     % Space above
  {}%                                     % Space below
  {\itshape}%                                     % Body font
  {}%                                     % Indent amount
  {\bfseries}%                            % Theorem head font
  {.}%                                    % Punctuation after theorem head
  { }%                                    % Space after theorem head, ' ', or \newline
  {\thmname{#1}\thmnumber{ #2}\thmnote{ (#3)}}%                                     % Theorem head spec (can be left empty, meaning `normal')

\theoremstyle{mystyle}
\usepackage{cite}
\usepackage{amsmath}
\usepackage{amssymb}
\newtheorem{thm}{Theorem}
\newtheorem{defn}{Definition}
\newtheorem{lem}{Lemma}

\newtheorem{cor}{Corollary}
\newtheorem{prop}{Proposition}

\newtheorem{rem}{Remark}
\newtheorem{assumption}{Assumption}
\newtheorem{problem}{Problem}
\newcommand{\sign}{{\rm sgn}}

\newcommand{\R}{\mathbb{R}}

\newcommand{\id}{\operatorname{id}}
\newcommand{\graph}{\mathcal{G}}
\newcommand{\nodes}{\mathcal{V}}
\newcommand{\edges}{\mathcal{E}}

\newcommand{\diag}{\operatorname{diag}}

\newcommand{\dsiso}{\textsc{DSISO}}

\newcommand{\ol}[1]{\overline{#1}}
\newcommand{\ul}[1]{\underline{#1}}

\newcommand{\real}{\mathbb{R}}

\newcommand{\Acompx}{\mathcal{A}_x}
\newcommand{\Acompd}{\mathcal{A}_d}
\newcommand{\Bcompx}{\mathcal{B}_x}
\newcommand{\Bcompd}{\mathcal{B}_d}

\definecolor{hotmagenta}{rgb}{1.0, 0.11, 0.81}

\usepackage{setspace}
\newcommand{\bulletsym}{\hbox{$\bullet$}}
\newcommand{\bulletend}{\relax\ifmmode\else\unskip\hfill\fi\bulletsym}

\renewcommand{\tilde}{\widetilde}

\newcommand{\mo}[1]{{\color{black} #1}}
\renewcommand{\sb}[1]{{\color{black} #1}}

\begin{document}

\title{Distributed Resilient Interval Observer Synthesis \\  for Nonlinear Discrete-Time Systems\vspace{-.0cm}}

\author{Mohammad~Khajenejad$^*$,~\IEEEmembership{Member,~IEEE,} Scott~Brown$^*$,~\IEEEmembership{Student Member,~IEEE}, and~Sonia~Martinez,~\IEEEmembership{Fellow,~IEEE}
  \thanks{$^*$Equal contribution.
    M. Khajenejad, S. Brown and S. Martinez are with the
    Department of Mechanical and Aerospace Engineering, University of
    California, San Diego, CA, USA (e-mail: mkhajenejad, sab007,
    soniamd@ucsd.edu). This work was partially supported by NSF grant 2003517,3, ONR grant N00014-19-1-2471, and ARL grant W911NF-23-2-0009.}}

\allowdisplaybreaks

\maketitle

\begin{abstract}
  This paper introduces a novel recursive distributed estimation
  algorithm aimed at synthesizing input and state interval observers
  for nonlinear bounded-error discrete-time multi-agent systems. The
  considered systems have sensors and actuators that are susceptible
  to unknown or adversarial inputs. To solve this problem, we first
  identify conditions that allow agents to obtain nonlinear
  bounded-error equations characterizing the input. Then, we
  propose a distributed interval-valued observer that is guaranteed to
  contain the disturbance and system states. \sb{\mo{To do this}, we
    first detail a gain design procedure that uses global problem data
    to minimize an upper bound on the $\ell_1$ norm of the observer
    error. We then \mo{propose} a gain design approach that does not
    require global information, using only values that are local to
    each agent. The second method improves on the computational
    tractability of the first, at the expense of some added
    conservatism. Further, we discuss some possible ways of extending
    the results to a broader class of systems. We conclude by
    demonstrating our observer on two examples. The first is a
    unicycle system, for which we apply the first gain design
    method. The second is a 145-bus power system, which showcases the
    benefits of the second method, due to the first approach being
    intractable \mo{for systems with high dimensional state spaces}.}

\end{abstract}
%\margin{I have rewritten the abstract. It still needs more
 % detail/precision}
%\begin{IEEEkeywords}
% mixed-monotone systems, %decomposition functions,
%constrained reachability, interval observer design,
%nonlinear systems
%\end{IEEEkeywords}

\IEEEpeerreviewmaketitle

\vspace{-0.5cm}
\section{Introduction}
\vspace{-0.05cm} \IEEEPARstart{T}{he} successful operation of
cyber-physical systems (CPS) relies on the seamless integration of
various computational, communication, and sensor components that
interact with each other and with the physical world in a complex
manner. CPS finds applications in diverse domains such as industrial
infrastructures~\cite{BC-JZ-GPH-SK-AWC:18}, power grids
\cite{JZ-AGE-MN-LM-AA-VT-IK-BP-AKS-JQ-ZH-APM:19}, \mo{autonomous
  vehicles} and intelligent
transportation systems \cite{YS-HS:17}. %\margin{autonomous
%  vehicles as well}
  In safety-critical systems like these, the
occurrence of unknown inputs, e.g., unstructured uncertainties,
unprecedented scenarios, and even malfunction or deliberate attacks by
malicious entities \cite{KZ:16} can lead to severe
consequences. In several forms of such occurrences, deceptive signals
are introduced into the actuator signals and sensor measurements by
strategic and/or malicious agents. These unknown inputs cannot be
accurately modeled as zero-mean Gaussian white noise or signals with
known bounds due to their strategic nature. Meanwhile, most
centralized algorithms for state estimation are computationally
intensive, particularly in realistic high-dimensional CPS
scenarios. Hence, the development of reliable distributed algorithms
for state and unknown input estimation becomes imperative to ensure
resilient control, unknown input reconstruction, and effective
mitigation strategies.

\emph{Literature review}.  Driven by the aforementioned
considerations, various estimation algorithms have been put forward,
aiming to address the challenge of jointly estimating the system state
and the unknown disturbance (input) through a central entity. For
instance, in \cite{GC-YZ-SG-WH:21}, the focus was on dealing with
unknown inputs/disturbances on actuators and sensors, while tackling
the secure state estimation and control problem, where the authors
proposed a $\chi^2$ detector to identify these malicious inputs. The
work in \cite{KD-XR-ASL-DEQ-LS:20} centered around remote state
estimation and the challenge of dealing with an active eavesdropper,
where to evaluate the stealthiness of the eavesdropper, the authors
presented a generalized framework and a criterion based on the packet
reception rate at the estimator. In \cite{CW-ZH-JL-LW:18}, a
sliding-mode observer was introduced to perform dual tasks: estimating
system states and identifying unknown inputs, simultaneously. On the
other hand, \mo{the research in \cite{EM-FY-QLH-LV:18,LL-LM-JG-JZ-YB:21}} proposed an
estimation approach based on projected sliding-mode observer to
reconstruct system states.
%\margin{I suggest we look for even more
 % recent references now}

Additionally, \mo{the work in \cite{MLC-AC:17,LL-WW-QM-KP-XL-LL-JL:21}} focused on
reconstructing input signals from the equivalent output injection
signal using a sliding-mode observer. In contrast, \mo{the procedure in \cite{AYL-GHY:17,PW-BC-LS-LY:23}}
treated an adversarial input as an auxiliary state and employed a
robust switching Luenberger observer, considering sparsity, to
estimate the state. %\margin{these seem older}

In scenarios where the noise signals follow Gaussian and white
characteristics, a substantial body of research has proposed diverse
methodologies, \mo{mainly based on extended Kalman filtering
  techniques,} for accomplishing joint input \mo{(or adversarial
  attack)} and state estimation. These methodologies include minimum
variance unbiased estimation \cite{SYZ-MZ-EF:16}, modified
double-model adaptive estimation \cite{PL-EJVK-CCDV-QC:16}, robust
regularized least square approaches \cite{MA-MR:18}, and
\mo{residual-based methods \cite{VR-BJG-JR-THS:21}}.
%\margin{fault detection and
 % identification techniques have been extensively used in the context
 % of adversarial attacks. These rely on Kalman filtering techniques
 % and estimating the fault (adversary). I think we should mention some
 % of these as well. }
  Nonetheless, since these algorithms assume
knowledge of uncertainty distribution, %\margin{that is not always the
%  case. Adversarial and stealthy detection techniques have been
 % proposed in the literature; e.g. the work by Justin Ruths and
%  myself.}
they are not applicable in the context of resilient bounded-error
worst-case estimation, where such information is unavailable.
% \margin{rather than this, I'd make the point that these
 % techniques do not consider the mitigation problem simultaneously}
 To tackle this issue, numerous techniques have been proposed for
\mo{ linear deterministic systems \cite{SL-SM-JC:23-tac}}, stochastic systems
\cite{HK-PG-MZ-PL:17}, and bounded-error systems
\cite{YN-YM:15,MP-PT-IL-GJP:15,SZY-MQF-EF:16}. Typically, these
methods yield point estimates, representing the most probable or
optimal single estimate, as opposed to set-valued estimates.

Set-valued estimates offer a valuable advantage by providing stringent
accuracy bounds, essential for ensuring safety
%\cite{MK-ZJ-SZY:21,MK-FS-SZY:21}
\cite{SZY:18,FB-MS:12,khajenejad2022resilient}.
%\margin{Are these
 % references just on safety? then I suggest we remove them ---in
%  addition is extremely bad practice to self cite so many times. And
 % as you can see, this is not my style at all, as there is no single
 % reference to my own work in this introduction (...) While I don't
 % oppose to include one or two self references, what we do here goes
 % well beyond what is acceptable practice}
  Additionally, employing
fixed-order set-valued methods can reduce the complexity of optimal
observers \cite{MM-AV:91,khajenejad2023resilientcdc}, which tends to grow over
time. Consequently, fixed-order centralized set-valued observers have
been introduced for various system classes
\cite{SZY:18,NE-DGD-DH:19,MK-SZY:22,TP-MK-SPD-SZY:22,MK-SZY:19a}.
%\margin{these
 % are way to many references in a row. Moreover, like 4 are Mohammad's,
%  one from yong. This is bad practice and frowned upon, just pick 3
 % references and explain what all these papers have (drawbacks).
%  Possible the 2022h from khajenejad (the latest one) is enough.}
These observers efficiently determine bounded sets of compatible
states and unknown inputs simultaneously. However, these algorithms
face challenges in scaling effectively within a networked setting,
particularly as the network size increases. This limitation has led to
the development of distributed input and state \sb{estimators}, which
primarily concentrate on systems with stochastic disturbances
\cite{YL-LZ-XM:13,AEA-AYK-FG:12}.
% \margin{there are a lot of
%   distributed consensus filters, are you going to mention these? if
%   these are on consensus, this is not a very representative type of
%   references, is it?}
% \margin{additionally, here we don't explain
%   well what is the challenge in making the previous observers
%   distributed. In my opinion the solutions rely on centralized
%   optimizations, and work is needed to develop decentralized synthesis
%   methods. Regarding the types of systems considered: is this done for
%   nonlinear systems as well?}
While these methods demonstrate superior scalability and robustness to
communication failures compared to their centralized counterparts,
they generally suffer from comparatively higher estimation
errors. Moreover, these methods are not applicable in bounded-error
settings where information about the stochastic characteristics of
noise or disturbance is unavailable. With this consideration, in our
previous work \cite{MK-SB-SM:22a-acc,MK-SB-SM:22b-acc}, we presented a
distributed algorithm for synthesizing interval observers for
bounded-error linear time-invariant (LTI) systems, without and with
unknown input signals, respectively. In this current study, our aim is
to extend our design presented in
\cite{MK-SB-SM:22a-acc,MK-SB-SM:22b-acc,TP-MK-SPD-SZY:22,khajenejad2022simultaneousijrnc} to address resiliency against
unknown inputs, in nonlinear bounded-error multi-agent settings.

\emph{Contribution.} This work aims to bridge the gap between
distributed resilient estimation algorithms and interval observer
design for scenarios with bounded errors and completely unknown and
distribution-free inputs for nonlinear multi-agent settings. To
achieve this:

1) We utilize a mixed-monotone decomposition of the nonlinear
dynamics, as well as \sb{a system transformation based on singular
  value decomposition (SVD)},
%\margin{what is this ``some system decompositions''?}
to rule out the effect of adversarial inputs and design resilient
observers.

2) \sb{We propose a four-step recursive distributed algorithm to
  design input and state observers of the system. The algorithm
  synthesizes interval-valued estimates for both states and unknown
  inputs. It utilizes the communication network to refine the
  individual set-valued estimates by taking the intersection of
  estimates among neighboring agents.}
  %\margin{never use ``etc'' in a paper. It
  %does not add anything, and gives a bad impression (are there so many
  %properties that we can't enumerate them all?) This last sentence is
 % way too long. Do we need the last part? probably not if we have
  %already stated it at the opening of the contributions
  %paragraph. This was in the abstract and I have removed it.}

% \margin{``some'' linear programs. Can we be more
%   precise than that? as you guessed, no I don't like the ``some''
%   either. Should this property be stated after the next item? This is
%   what we need for method 2.}

3) We establish two novel tractable alternative designs for ensuring
stability of our proposed observer, which are proven to minimize an
upper bound for the interval widths of observer errors. The first
method, which requires central knowledge of all system parameters,
takes the form of a mixed-integer linear program (MILP). \mo{However,
  these MILPs are not computationally efficient in systems with high
  state dimensions or networks with many agents. This motivates
  proposing} the second and more tractable procedure \mo{that} reduces
\mo{the} large MILP into many smaller optimization problems, which may
be solved much more efficiently at the cost of some additional
conservatism. \sb{For this we utilize the concept of ``collective
  positive detectability over neighborhoods'' (CPDN). We show that the
  CPDN property holds for a broad range of nonlinear multi-agent
  systems and can be verified by solving a linear program for each
  agent.}
%\margin{the
  %tractable approaches rely on the stability theorem 1, it deserves a
  %especial mention. It is said the second procedure is more efficient,
 % why? explain the type of problems we get}

% 5) The observer design is tractable and computationally efficient,
% making it a valuable approach to tackle these challenging estimation
% scenarios, which \mo{was} illustrated via simulations and comparisons with
% some benchmark observers.
4) \sb{We illustrate our algorithms' performance via two simulation
  examples and a comparison with an existing distributed interval
  observer.}  In particular, we considered a low-dimensional unicycle
dynamics, for which the first proposed method successfully returns
stable and optimal gains, while the second design is unable to find
feasible gains. Further, we consider a high-dimensional power system
example. In this case, the MILP-based first method becomes
intractable, while our second design returns stabilizing gains in a
reasonable time. This demonstrates that each approach may yield good
results on a case-by-case basis, with an \sb{intuitive tradeoff
  between conservatism and tractability.}
%\margin{What simulations and what
 % comparisons. There should be a mention to the discussion of how the work is extended. }

\section{Preliminaries}
\emph{{Notation}.} The symbols $\R^n$, $\R^{n \times p}$,
$\mathbb{N}$, $\mathbb{N}_n$, $\R_{\geq 0}$ and $\R_{>0}$ denote the
$n$-dimensional Euclidean space, the sets of $n$ by $p$ matrices,
natural numbers (including 0), natural numbers from 1 to $n$,
non-negative real, and positive real numbers, respectively. The
Euclidean norm of a vector $x \in \R^n$ is denoted by
$\|x\|_2\triangleq \sqrt{x^\top x}$.  For $M \in \R^{n \times p}$,
$M_{ij}$ denotes $M$'s entry in the $i$'th row and the $j$'th column,
$M^{\oplus}\triangleq \max(M,\mathbf{0}_{n,p})$,
$M^{\ominus}=M^{\oplus}-M$ and $|M|\triangleq M^{\oplus}+M^{\ominus}$,
where $\mathbf{0}_{n,p}$ is the zero matrix in $\R^{n \times p}$. The
element-wise sign of $M$ is $\sign(M) \in \R^{n \times p}$ with
$\sign(M_{ij})=1$ if $M_{ij} \geq 0$ and $\sign(M_{ij})=-1$,
otherwise. We use the notation $(M)_s$ to denote the row vector
corresponding to the $s$\textsuperscript{th} row of $M$. For vectors
in $\real^n$, the comparisons $>$ and $<$ are considered
element-wise. Finally, an interval
$\mathcal{I} \triangleq [\ul{z},\ol{z}] \subset \R^n$ is the set of
all real vectors $z \in \R^{n}$ that satisfies
$\ul{z} \le z \le \ol{z}$, with interval width
$\|\ol{z}-\ul{z}\|_{\infty} \triangleq \max_{i \in
  \{1,\cdots,{n_z}\}}{|\ol{z}_i-\ul{z}_i|}$.  Next, we introduce some definitions
and related results that will be useful throughout the paper.  First,
we review some mixed-monotonicity theory basics that will be leveraged
in our interval observer design.

\begin{defn}[Jacobian Sign-Stable {\cite[Definition 1]{khajenejadtight23}}]
  A function $f : \mathcal{Z} \subset \R^n \to \R^p$ is \emph{Jacobian
    sign-stable} (JSS) if the sign of each element of the Jacobian
  matrix does not change over the domain $\mathcal{Z}$. In other
  words, $J^\mu_{ij}(z) \geq 0$ or $J^\mu_{ij}(z) \leq 0,$
  $\forall z \in \mathcal{Z}.$ %\bulletend
\end{defn}

\begin{prop}[Jacobian Sign-Stable (JSS) Decomposition
  {\cite[Proposition 2]{MK-FS-SZY:22a}}]
  \label{prop:JSS_decomp}
  If a mapping $f: \mathcal{Z} \subset \R^n \to \R^p$ has Jacobian
  matrices satisfying $J^f(x) \in [\ul{J}^f,\ol{J}^f]$,
  $\forall z \in \mathcal{Z}$, where
  $\ul{J}^f,\ol{J}^f \in \R^{p \times n}$ are known bounds, then
  the mapping $f$ can be decomposed into an additive-remainder form:
  \begin{align}\label{eq:JSS_decomp}
     f(z) = H z + \mu(z), \;\forall z \in \mathcal{Z},
  \end{align}
  where the matrix $H\in\R^{p \times n}$ satisfies
  \begin{align}\label{eq:H_decomp}
    H_{ij} = \ul{J}^f_{ij} \ \text{ or } H_{ij} = \ol{J}^f_{ij} ,\;  \forall (i,j) \in \mathbb{N}_p \times \mathbb{N}_{n},
  \end{align}
  % \margin{Is this decomposition unique? how do we decide how each
  % component is saturated?}
  and the function $\mu$ is Jacobian sign-stable.
  % Scott: TODO waiting to address in simulation section
  \bulletend
  % \margin{I've changed the last $J^\nu_{ij}(z) $ to $J^\mu_{ij}(z)$
  %   Question: can the signs be different for each $ij$? Let's put this
  %   outside as a definition; as we use it again in prop 2.}
  % Scott: moved JSS to its own def.
\end{prop}

\begin{defn}[Mixed-Monotone Decomposition Functions]
  \cite[Definition 4]{LY-OM-NO:19}
  \label{defn:dec_func}
  Consider a function $g : \mathcal{X} \subset \R^{n} \to \R^{n}$.  A
  function $g_d: {\mathcal{X} \times \mathcal{X}} \to \R^{n}$ is a
  mixed-monotone decomposition function for $g$ if it satisfies \mo{the following conditions:}
  \begin{enumerate}
  \item $g_d({x,x})=g({x})$,
  \item $g_d$ is monotonically increasing in its first argument,
  \item $g_d$ is monotonically decreasing in its second argument.
  \end{enumerate}
  %\bulletend
\end{defn}

\begin{prop}[Tight and Tractable Decomposition Functions for JSS Mappings]
  \cite[Proposition 4 \& Lemma 3]{MK-FS-SZY:22a}
  \label{prop:tight_decomp}
  Suppose $\mu: \mathcal{Z} \subset \R^{n} \to \R^p$ is a JSS
  mapping. Then, for each $\mu_i$, $i \in \mathbb{N}_p$, a
  mixed-monotone decomposition function is given by
  \begin{align}\label{eq:JJ_decomp}
    \mu_{d,i}(z_1,z_2) \triangleq \mu_i(D^i z_1 + (I_{n} - D^i) z_2),
\end{align}
for any $z_1, z_2 \in \mathcal{Z}$ which satisfy either $z_1 \ge z_2$
or $z_1 \le z_2$,
%\margin{what is an ordered $z_1, z_2$?, according to
% what order? this has to go to the notations}
% Scott: resolved, we already say it's element-wise
\begin{align}\label{eq:Dj}
  D^i = \mathrm{diag}(\max(\mathrm{sgn}(\ol{J}^{\mu}_i),\mathbf{0}_{1,{n_z}})).
\end{align}
Moreover, assume that $\mu$ is the additive remainder in a JSS
decomposition of a function $f$ as in Proposition
\ref{prop:JSS_decomp}. Then, for any interval domain
$\ul{z} \leq z \leq \ol{z}$ of $f$, with
$z,\ul{z},\ol{z} \in \mathcal{Z}$ and
$\varepsilon \triangleq \ol{z}-\ul{z}$, the following inequality
holds:
\begin{align}
  \delta^{\mu}_d \leq \ol{F}_{\mu}\varepsilon, \ \text{where} \
  \ol{F}_{\mu}\triangleq \ol{J}^\oplus_f + \ul{J}_f^\ominus.
\end{align}
where $\delta^\mu_d \triangleq \|\mu_d(\ul{z},\ol{z}) - \mu_d(\ol{z}, \ul{z})\|_\infty$ \bulletend %$\ol{F}_{\mu}\triangleq\max(\ol{J}_f,\mathbf{0}_{p,n_z})-\min{(\ul{J}_f,\mathbf{0}_{p,n_z})}$.
\end{prop}
% \margin{what is $\delta^\mu_d$ is?}
% Scott: resolved
Consequently, by applying Proposition \ref{prop:tight_decomp} to the
Jacobian sign-stable decomposition obtained using Proposition
\ref{prop:JSS_decomp}, a tight and tractable decomposition function
can be obtained. Further details can be found in
\cite{MK-FS-SZY:22a}. %\margin{why? }

Finally, we recap a very well-known result in the literature, that
will be frequently used throughout the paper.
\begin{prop}\cite[Lemma 1]{DE-TR-SC-AZ:13}\label{prop:bounding}
  Let $A \in \real^{p \times n}$ and
  $\ul{x} \leq x \leq \ol{x} \in \real^n$. Then,
  $A^+\ul{x}-A^{-}\ol{x} \leq Ax \leq
  A^+\ol{x}-A^{-}\ul{x}$. As a corollary, if $A$ is
  non-negative, $A\ul{x} \leq Ax \leq A\ol{x}$. \bulletend
\end{prop}

\section{Problem Formulation}
Consider a multi-agent system (MAS) consisting of $N$ agents, which
interact over a time-invariant communication network represented as a
graph $\graph = (\nodes, \edges)$. The agents are able to obtain individual
measurements of a target system as described by the following
nonlinear dynamics:
\begin{align}\label{eq:ind_system}
  \begin{split}
    x_{k+1} &= f(x_k,w_k) + G d_k,\\
    y^i_k &= C^i x_k + D^i v^i_k + H^i d_k, \ i \in \mathcal{V}, \ k \in \mathbb{Z}_{\geq 0},
  \end{split}
\end{align}
with state $x_k \in \mathcal{X} \subset \R^n$, outputs
$y^i_k \in \R^{l_i}$, unknown input $d_k \in \R^p$, and bounded
disturbances $w_k \in [\ul{w},\ol{w}] \subset \R^{n_w}$ and
$v^i_k \in [\ul{v}^i,\ol{v}^i] \subset \R^{n^v_i}$. We assume the
function $f$ and matrices $G$, $C^i$, $D^i$, and $H^i$ are known and
have compatible dimensions. Unless otherwise noted, a superscript $i$
means an object is associated with node $i$.
% \margin{say that all mappings and matrices have the consistent
% dimensions, and we have to say these are known, right?}
% Scott: resolved
% \margin{expand the rank operator, so that it is clear what it is}
% Scott: resolved

% \margin{With the remark below, we don't need the last sentence above.}
% Scott: resolved
\noindent \textbf{\emph{Unknown Input Signal Assumptions.}}  The
unknown inputs $d_k$ are not constrained to follow any model nor to be
a signal of any type (random or strategic).  We also do not assume
that $d_k$ is bounded. In other words no prior useful knowledge of the
nature of $d_k$ is available. Therefore $d_k$ is suitable for
representing scenarios including adversarial attack signals, a unknown
entity operating a target vehicle, and more.

Moreover, we assume the following, which is satisfied for a broad
range of nonlinear functions \cite{LY-NO:19}:

\begin{assumption}\label{assumption:mix-lip}
  The vector field $f$ has a bounded Jacobian over the domain
  $\mathcal{X} \times \mathcal{W}$, i.e., for all
  $(x, w) \in \mathcal{X} \times \mathcal{W}$,
  \begin{gather*}
    J^f_x(x, w) \in [\ul{J}^f_x,\ol{J}^f_x] \text{ and }
    J^f_w(x, w) \in [\ul{J}^f_w,\ol{J}^f_w].
  \end{gather*}
  The Jacobian bounds $\ul{J}^f_x$, $\ol{J}^f_x$, $\ul{J}^f_w$, and
  $\ol{J}^f_w$ are known. \bulletend
\end{assumption}

% Scott: isn't this redundant with Problem 1?
The MAS's goal is to estimate the trajectories of the plant in
\eqref{eq:ind_system} in a distributed manner, see
Problem~\ref{prob:SISIO}. % %when
%they are initialized in a given interval
%$\interval_x \triangleq [\ul{x}_0,\ol{x}_0] \subset \R^{n}$, with
%$\ul{x}_0$ and $\ol{x}_0$ known to all agents.
%Further, we formally define
The formal statement of the problem relies on the notions of
\textit{framers}, \textit{correctness} and \textit{stability}, which
are defined next.

\begin{defn}[Correct Individual Framers]
  For an agent $i\in\nodes$ the sequences $\{\ol{x}^i_k\}_{k \ge 0}$
  and $\{\ul{x}^i_k\}_{k \ge 0} \in \R^{n}$ are called \textit{upper}
  and \textit{lower} individual state framers for \eqref{eq:ind_system} if
  \begin{align*}
    \ul{x}^i_k \leq x_k \leq \ol{x}^i_k, \quad \forall k \ge 0.
  \end{align*}
  Similarly, $\{\ol{d}^i_k\}_{k \ge 0}$ and
  $\{\ul{d}^i_k\}_{k \ge 0} \in \R^{n}$ are input
  framers for \eqref{eq:ind_system}, if
  \begin{align*}
    \ul{d}^i_k \leq d_k \leq \ol{d}^i_k, \quad \forall k \ge 0.
  \end{align*}
  Also, we define
  \begin{align}
    \label{eq:error_1}
    {e}^i_{x,k} \triangleq \ol{x}^i_k-\ul{x}^i_k, \quad
    {e}^i_{d,k} \triangleq \ol{d}^i_k-\ul{d}^i_k, \quad \forall k \ge 0,
  \end{align}
  the individual state and input framer errors, respectively. \bulletend
\end{defn}
\begin{defn}[Distributed Resilient Interval Framer]
  \label{defn:dist_obs}
  For an MAS with target System \eqref{eq:ind_system} and
  communication graph $\graph$, a distributed resilient interval
  framer is a distributed algorithm over $\graph$ that allows each
  agent in a MAS to cooperatively compute individual correct upper and
  lower state and input framers, for any arbitrary realization of the
  unknown input (attack) sequence. \bulletend
\end{defn}

\begin{defn}[Collective Framer Error]
  \label{defn:framer-error}
  For a distributed interval framer, the collective framer state and
  input errors % $e_{x,k}$ and $e_{d,k}$ are
  are the vectors
  \begin{align}\label{eq:error_2}
    \begin{array}{rl}
      e_{x,k} &\triangleq
                \begin{bmatrix}
                  ({e}^1_{x,k})^\top & \cdots & ({e}^N_{x,k})^\top
                \end{bmatrix}^\top \in \R^{Nn},\\
      e_{d,k} &\triangleq
                \begin{bmatrix}
                  ({e}^1_{d,k})^\top & \cdots & ({e}^N_{d,k})^\top
                \end{bmatrix}^\top \in \R^{Np}.
    \end{array}
  \end{align}
  of all individual lower and upper state and input framer errors,
  respectively. \bulletend
\end{defn}

\begin{defn}[Collective Input-to-Sate Stable (C-ISS) Distributed Resilient Interval Observer]
  \label{defn:stability}
  % \margin{since we refer later to it in this way, I have changed this
  % title accordingly}
  % Scott: change stability -> stable
  A distributed resilient interval framer is
  collectively input-to-state stable (C-ISS), if the collective state
  framer error (cf.~Definition~\ref{defn:framer-error}) satisfies:
  \begin{align*}
    \|e_{x,k}\|_2 \leq
    \beta(\|e_{x,0}\|_2,k)+\rho\big(\max_{0\le l \le k}|\Delta_l|\big), \quad \forall k \in
    \mathbb{Z}_{\geq 0},
  \end{align*}
  where
  $\Delta_l \triangleq [ w_l^\top \ v_l^{1\top} \cdots
  v_l^{N\top}]^\top\in \R^{n_w + Nn_v}$, $\beta$ and $\rho$ are
  functions of classes $\mathcal{KL}$ and
  $\mathcal{K}_{\infty}$~\cite{HKK:02}
  respectively.
  %\margin{It's probably not necessary to add this in a
  % footnote. At most a reference to Khalil's book?}
  % Scott: resolved
  % We call a C-ISS,
  %distributed resilient interval framer a
  In this case, the framer is referred to as a \textbf{C-ISS
    distributed resilient interval observer}.
  % When the multi-agent
  % system consists of a single agent, the notion of C-ISS distributed
  % resilient interval framer reduces to the notion of an
  % input-to-state-stable (ISS) resilient framer for a centralized
  % resilient interval observer.
  \bulletend
  % \margin{can we just call it
  %   a C-ISS distributed resilient interval framer. (period)? }
  % Scott: removed the last sentence above
\end{defn}

The resilient observer design problem is stated next:
% Scott: I thought we took out the consensus, since we don't achieve
% it \margin{I'd like to bring back the comment of Scott's: are we
% including consensus here or not? If it is not necessary, we don't
% need to introduce these variables $\xi$ in the statement either. If
% we don't require consensus, how can agents have a common picture of
% the estimate of a target? Having a common view is a usual
% requirement in MAS, so this requires an explanation }

% Scott: removed that part of the problem statement. I'm working on
% adding a bound on the disagreement between nodes. I would say that
% in general, exact consensus is not the goal of a distributed
% estimation algorithm. The primary goal is to reduce the uncertainty
% of every agent's estimate. This is accomplished using consensus-like
% operations to spread information through the network without
% all-to-all communication.  This helps to achieve the "common view",
% but the nodes will not have *exactly* the same estimates.

% \margin{what is the end goal of the paper? the design of a framer or of an observer? If we achieve the harder problem, which is the observer design, then, we don't have to mention the framer design here. Then, in the following section, we can say we start by constructing a framer first, no problem}
% Scott: The end goal is the observer, so we can say that here.
\begin{problem}\label{prob:SISIO}
  Given an MAS and the uncertain nonlinear system
  in~\eqref{eq:ind_system}, design a %design a C-ISS distributed resilient framer, i.e., a
  distributed resilient interval observer. \bulletend
\end{problem}

\section{Distributed Interval Framer Design}
% \margin{I think this title should refer to ``framer design'', not ``observer'' design}
% Scott: resolved

\label{sec:observer}
In this section, we describe the structure of our proposed distributed
resilient interval framer, as well as its correctness. This lays the
groundwork for the computation of stabilizing observer gains, which is
discussed in the following section.
% \margin{Mention that in a subsequent section }
% Scott: resolved, maybe revisit/edit wording

Our strategy for synthesizing a distributed resilient interval framer
in the presence of unknown inputs consists of a preliminary step and a
recursive observer design.  First, in Section~\ref{sec:transform},
each agent obtains an equivalent representation of the system which
uses output feedback to remove the attack signal from the system
dynamics. After this transformation, each agent performs the four
steps described in Section~\ref{sec:framer} to compute state and input
framers at every time step.
\subsection{Preliminary System Transformation}
\label{sec:transform}
First, we briefly introduce a system transformation similar to that
used in \cite{SZY:18,khajenejad2020simultaneousfullrank,khajenejad2021simultaneousdeficient}, which will enable computation of
state framers despite the presence of the unknown input. The following
paragraphs describe the transformation that is performed for every
agent $i \in \nodes$.

Let $r^i\triangleq {\rm rank} (H^i)$. By applying a singular value
decomposition, we have
\begin{gather*}
  H^i= \begin{bmatrix}U^i_1& U^i_2 \end{bmatrix}
  \begin{bmatrix} \Xi^i & 0 \\ 0 & 0 \end{bmatrix}
  \begin{bmatrix} V_1^{i, \top}, \\
    V_2^{i, \top} \end{bmatrix}
\end{gather*}
with $V^i_1 \in \R^{p \times r^i}$,
$V^i_2 \in \R^{p \times (p-r^i)}$,
$\Xi^i \in \R^{r^i \times r^i}$ (a diagonal matrix of
full rank),
$U^i_1 \in \R^{l^i \times r^i}$ and
$U^i_2 \in \R^{l^i \times (l^i-r^i)}$.  Then, since
$V^i\triangleq \begin{bmatrix} V^i_1 & V^i_2 \end{bmatrix}$ is
unitary,
\begin{align}
  \label{eq:d12}
  d_k =V^i_1 d^i_{1,k}+V^i_2 d^i_{2,k}, \ d^i_{1,k}=V_1^{i,\top} d_k, \ d^i_{2,k}=V_2^{i,\top} d_k.
\end{align}
By means of these, the output equation can be decoupled, and the agent
can obtain an equivalent representation of the target state equation
and its own measurement equation.
\begin{subequations}
  \label{eq:stateq}
  \begin{align}
    x_{k+1}&=f(x_k,w_k)+G^i_1
             d^i_{1,k}+G^i_2d^i_{2,k},\\
    z^i_{1,k}&= C^i_1x_k +  D^i_1v^i_{1,k}+\Xi^i d^i_{1,k}, \label{eq:z1k}\\
    z^i_{2,k}&= {C^i_2x_k + D^i_2v^i_{2,k}},  \label{eq:z2k}\\
    d_k &=V^i_1 d^i_{1,k}+V^i_2 d^i_{2,k}, \label{eq:dkk}
  \end{align}
\end{subequations}
where
\begin{gather*}
  \begin{bmatrix}
    C^i_1 \\ C^i_2
  \end{bmatrix} = (U^i)^\top C, \quad
  \begin{bmatrix}
    G^i_1 \\ G^i_2
  \end{bmatrix} = (U^i)^\top G, \ \text{and} \
  \begin{bmatrix}
    D^i_1 \\ D^i_2
  \end{bmatrix} = (U^i)^\top D.
\end{gather*}

% \marginsc{need to write something about this assumption}
% \margin{Indeed, and probably we need to explain why the assumptions that lead to this transformation are reasonable. Note that no single agent can estimate the target on its own, no?}
% Scott: added new sentence and remark
Finally, we make an assumption that ensures that every agent is able
to obtain bounded estimates of the unknown input. We refer the reader
to~\cite{SZY:18} for a discussion of the necessity of this assumption
in obtaining bounded estimates.
\begin{assumption}\label{ass:rank}
   $C^i_2G^i_2$ has full column rank for all $i \in \nodes$. Hence, there exists
   $M^i_2 \triangleq (C^i_2G^i_2)^\dagger$, such that
   $M^i_2C^i_2G^i_2=I$.
 %  \bulletend
\end{assumption}
\begin{rem}
  It is not strictly necessary that Assumption~\ref{ass:rank} is
  satisfied for all $i$. By utilizing another SVD, it is possible that
  nodes can obtain an estimate of a partial component of $d^i_{2,k}$,
  relying on neighbors to estimate the other components.  These
  details, though straightforward in practice, complicate the
  exposition significantly, so we proceed with
  Assumption~\ref{ass:rank} for the sake of simplicity.
%  \bulletend
\end{rem}
\subsection{Interval Framer Design}
\label{sec:framer}
% \margin{We should start with this explanation at the beginning of the previous section}
% Scott: resolved
Having performed the system transformation in the previous section, we
can now describe the design of the interval framer, which is a
four-step recursive process. Inspired by our previous work on
synthesizing interval observers for nonlinear systems
\cite{MK-FS-SZY:22a,MK-SZY:22,khajenejad2021intervalmodel},
% \margin{are these papers for nonlinear or for linear systems?}
% Scott: nonlinear, TODO fix bib
each agent designs local interval framers for the equivalent system
representation, which returns local state framers
(Step~\hyperlink{step1}{i}).
% \margin{Refer to steps with a label as Step i), etc}
% Scott: not sure how to do this properly, but this should work
Next, agents share their local interval state estimates with their
neighboring agents and update their estimates by taking the best
estimates via intersection (Step~\hyperlink{step2}{ii)}). Then, each
agents compute their local input framers as functions of the updated
state framers (Step~\hyperlink{step3}{iii)}). Finally, agents update
their local interval input estimates via intersection
(Step~\hyperlink{step4}{iv)}).

The following lemma formalizes the preliminary step.
\begin{lem}[Equivalent System Representation]
  \label{lem:equiv}
  Suppose Assumptions \ref{assumption:mix-lip} and \ref{ass:rank}
  hold. Then, System \eqref{eq:ind_system}, and equivalently the MAS
  in \eqref{eq:stateq},  admits the following representation%can be represented as follows:
  \begin{align}
    \label{eq:gamma_dyn_2}
    \begin{split}
    x_{k+1} &= (T^iA^i-L^iC^i_2)x_k
              + T^i\rho^i(x_k,w_k) \\
            &\quad + \Psi^i\eta^i_{k+1}
              + {\zeta}^i_{k+1},
    \end{split} \\
    \label{eq:di1}
    d^i_{1,k} &= M^i_1(z^i_{1,k} -
                C^i_1x_k -
                D^i_1v^i_k),\\
    \label{eq:di2}
    \begin{split}
    d^i_{2,k} &= M^i_2 C^i_2(G^i_1M^i_1C^i_1x_k -
                f(x_k,w_k)) \\
            &\quad + M^i_2 C^i_2 G^i_1 M^i_1 D^i_1 v^i_k
              - M^i_2 D^i_2 v^i_{k+1} \\
            &\quad - M^i_2 C^i_2 G^i_1 M^i_1 z^i_{1,k}
              + M^i_2 z^i_{2,k+1}.
    \end{split}
  \end{align}
  Here, $T^i,\Gamma^i$ and $L^i$ are (free-to-choose) matrices of
  appropriate dimensions, which are constrained by
  \begin{align}
    \label{eq:gain_con}
    T^i = I - \Gamma^i C^i_2.
  \end{align}
  See Appendix~\ref{sec:lem1_mat} for an explicit expression of
  $M^i_1$, $M^i_2$, $\Phi^i$, $\Psi^i$, $\eta^i_{k+1}$, and
  $\zeta^i_{k+1}$. The matrix $A^i$ and the JSS mapping $\rho^i$ are
  obtained by applying Proposition~\ref{prop:JSS_decomp} to the vector
  fields $f^i$:
  \begin{align}
    {f}^i(x,w) &\triangleq f(x,w)-\Phi^i C^i_1x.
  \end{align}
\end{lem}
\begin{proof}
  The proof is given in Appendix \ref{sec:lem1_proof}.
\end{proof}
Note that the observer gains $T^i,\Gamma^i$ and $L^i$ will be designed
later (cf.~Section~\ref{sec:synthesis}) to ensure stability and optimality of the proposed observer.

After deriving the equivalent system representation in~
\eqref{eq:gamma_dyn_2}--\eqref{eq:di2}, subject to~\ref{eq:gain_con}
% \margin{add: subject to (14)}
% Scott: resolved
local state and input framers can be constructed. Then, by leveraging
the network structure, the local framers will be refined by choosing
the best of framers among neighboring agents. This results in a
recursive four-step distributed framer design that can be summarized
as follows.

\noindent\textbf{\textit{ \hypertarget{step1}{Step i)} State Propagation and Measurement
    Update:}}

Applying Proposition \ref{prop:bounding} to bound the linear terms
(with respect to state and/or noise), as well as leveraging tight
decomposition functions given by $\rho^i_d$ (cf. Proposition
\ref{prop:tight_decomp}) for the nonlinear components $\rho^i$ in
\eqref{eq:gamma_dyn_2}, we obtain the following dynamical system.  By
construction, the system is guaranteed to bound the true state values
of \eqref{eq:gamma_dyn_2},
% \margin{is this a statement that is later proven in Lemma 3?. It seems
%  the proof is right here, so do we really need Lemma 3?}
% Scott: no change necessary
and therefore, it returns local state framers for
\eqref{eq:ind_system}:
%\margin{I would say for (10). Can we remove the tilde from this stiff matrix?}
% Scott: TODO
\begin{align}
  \label{eq:framers}
  \begin{bmatrix} \ol{x}^{i,0}_k \\ \ul{x}^{i,0}_k \end{bmatrix}
  = \mathtt{\tilde{A}}^i \begin{bmatrix} \ol{x}^i_k \\ \ul{x}^i_k \end{bmatrix}
  + \mathtt{T}^i
  \begin{bmatrix} \rho^i_d(\ol{x}^i_k,\ol{w},\ul{x}^i_k,\ul{w}) \\
    \rho^i_d(\ul{x}^i_k,\ul{w},\ol{x}^i_k,\ol{w}) \end{bmatrix}
  + \mathtt{\Psi}^i\begin{bmatrix} \ol{\eta}^i \\ \ul{\eta}^i \end{bmatrix}
  + \zeta^i_{k+1},
\end{align}
with $\mathtt{\Psi}^i$, $\mathtt{T}^i$, $\mathtt{\tilde{A}}^i$,
$\ol{\eta}^i$, $\ul{\eta}^i$, and $\zeta^i_{k+1}$ given in Appendix
\ref{sec:framer_mat}.

\noindent\textbf{\textit{ \hypertarget{step2}{Step ii)} State Framer Network Update:}}

Given the previous framers, each agent~$i$
will iteratively share its local interval estimate
with its neighbors in the network, updating them by taking the
tightest interval from all neighbors via intersection:
\begin{align}\label{eq:network_update}
    \ul{x}^{i}_{k} = \max_{j\in\mathcal{N}_i}\ul{x}^{j,0}_k, \quad
    \ol{x}^{i}_{k} = \min_{j\in\mathcal{N}_i}\ol{x}^{j,0}_k,
\end{align}
A schematic of the intersection-based network update is shown in
Figure \ref{fig:intersection}.
\begin{figure}
  \begin{center}
    \includegraphics[]{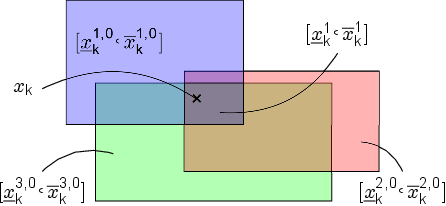}
  \end{center}
  \caption{\small{A schematic of the intersection-based network update step.}}
  \label{fig:intersection}
\end{figure}

\noindent\textbf{\textit{ \hypertarget{step3}{Step iii)} Input Estimation:}}

Next, it is straightforward to see that, plugging $d^i_{1,k}$ and
$d^i_{2,k}$ from \eqref{eq:di1} and \eqref{eq:di2} into
\eqref{eq:dkk}, returns
% \margin{I'm rewriting this part to organize it better.}
% Scott: looks fine, made some minor corrections
\begin{align}
  \label{eq:d_x}
  \begin{split}
    d_k &= h^i(x_k,w_k) + \Upsilon^i D^i_1 v^i_k
          + \Theta^i D^i_2 v^i_{k+1} + \zeta^i_{d,k+1},% \\
     %  &= A^i_h x_k + \mu^i(x_k,w_k) + \Lambda^i \eta^i_k,
  \end{split}
\end{align}
for appropriate $\zeta^i_{d,k+1}$, $\Theta^i$, and $\Upsilon^i$ given in Appendix~\ref{sec:input_mat} and
\begin{align*}
  h^i(x,w) &\triangleq \Upsilon^i C^i_1 x
             + \Theta^i C^i_2 f(x,w).
            % = A^i_h x + C^i_h w + \mu^i(x,w),
\end{align*}
By Proposition \ref{prop:JSS_decomp} there are matrices $A^i_h$,
$C_h^i$ and a vector field $\mu^i$ that result in the JSS
decomposition of $h^i(x,w) = A^i_h x + C^i_h w + \mu^i(x,w)$,
which leads to $d_k = A^i_h x_k + \mu^i(x_k,w_k) + \Lambda^i \eta^i_k$,
for appropriate variables $\eta_k^i$; see
Appendix~\ref{sec:input_mat}.
% \margin{as far as I know, I haven't seen these variables defined
% anywhere. I suppose they match the previous $\underline{\eta}$
% $\overline{\eta}$}
% Scott: they are in A1. Need to discuss broader problem of variable definitions

% \margin{In the same way that
% we have stated that the dynamical system in (16) is guaranteed to
% bound the true state values, we could say here that (19) also
% guarantees to bound the d values. However, this seems to be part of
% the statement of lemma 3}
% Scott: added a sentence below
%\begin{align}
%  \label{eq:d_x}
%  \begin{split}
%    d_k &= h^i(x_k,w_k) + \Upsilon^i D^i_1 v^i_k
%          + \Theta^i D^i_2 v^i_{k+1} + \zeta^i_{d,k+1}, \\
%       &= A^i_h x_k + \mu^i(x_k,w_k) + \Lambda^i \eta^i_k,
%  \end{split}
%\end{align}
%where
%\begin{align*}
%  h^i(x,w) &\triangleq \Upsilon^i C^i_1 x

%+ \Theta^i C^i_2 f(x,w),
            % = A^i_h x + C^i_h w + \mu^i(x,w),

%\end{align*}
%and, we are making use of the existing matrices $A^i_h$, $C_h^i$ and
% vector field $\mu^i$ that result in the JSS decomposition
% $h^i(x,w) = A^i_h x + C^i_h w + \mu^i(x,w)$; see Proposition
% \ref{prop:JSS_decomp}. The matrices $\Upsilon^i ,\Lambda^i$ and
% variables $ \zeta^i_{d,k+1}$ are defined in Appendix
% \ref{sec:input_mat}. \margin{The variables $\eta^i_k$ have not been
% defined so far in this section (unlike their upper and lower bounds
% given above.)}

Applying Propositions \ref{prop:tight_decomp} and \ref{prop:bounding}
to \eqref{eq:d_x}, yields
% \margin{Wouldn't it better to say we apply
%   proposition 3 to (16) and (19) respectively? We already said the
%   state values of (11) and of (18) are guaranteed to bound these
%   values in (16) and (19)}
\begin{align}\label{eq:input_framers}
  \begin{bmatrix} \ol{d}^{i,0}_k \\ \ul{d}^{i,0}_k \end{bmatrix}
  = \mathtt{A}^i_h
  \begin{bmatrix} \ol{x}_k \\ \ul{x}_k \end{bmatrix} +
  \begin{bmatrix}
    \mu^i_d(\ol{x}_k,\ol{w},\ul{x}_k,\ul{w}) \\
    \mu^i_d(\ul{x}_k,\ul{w},\ol{x}_k,\ol{w})
  \end{bmatrix}
  + \mathtt{\Lambda}^i
  \begin{bmatrix} \ol{\eta}^i \\ \ul{\eta}^i \end{bmatrix}
  + \zeta^i_{d,k+1},
\end{align}
where $\mu^i_d$ is the tight decomposition of $\mu^i$, and
$\mathtt{A}^i_h$, $\mathtt{\Lambda}^i$ are given in
Appendix~\ref{sec:input_mat2}. Again, this expression is guaranteed to
bound the true value of $d_k$ by construction.

\noindent\textbf{\textit{ \hypertarget{step4}{Step iv)} Input Framer Network Update:}}

Finally, similar to Step ii), each agent~$i$ shares its local input
framers with its neighbors in the network, again taking the
intersection:
\begin{align}
  \label{eq:network_update_d}
  \begin{split}
    \ul{d}^{i}_{k} = \max_{j\in\mathcal{N}_i}\ul{d}^{j,0}_k, \quad
    \ol{d}^{i}_{k} = \min_{j\in\mathcal{N}_i}\ol{d}^{j,0}_k.
  \end{split}
\end{align}
\begin{prop}\label{lem:error_diameter}
  Given the neighbors' state and input interval estimates
  $\{\ul{x}^{j,0}_k, \ \ol{x}^{j,0}_k
  \}_{j\in\mathcal{N}_i}$ and
  $\{\ul{d}^{j,0}_k, \ \ol{d}^{j,0}_k
  \}_{j\in\mathcal{N}_i}$, \eqref{eq:network_update} and
  \eqref{eq:network_update_d} result in the smallest possible state
  and input intervals (i.e., the ones with the smallest width in all
  dimensions), which are guaranteed to contain the true state and
  input, respectively.
\end{prop}
\begin{proof}
  The statement follows from the definition of the intersection of
  intervals.
\end{proof}
An important consequence of Proposition \ref{lem:error_diameter} is that our
observer is guaranteed to perform better than one which uses a linear
operation (i.e., averaging) to communicate across the network.
Despite the nonlinearity of \eqref{eq:network_update}, we are still
able to provide a thorough stability analysis, which is a key
contribution of this work. Figure \ref{fig:cons} illustrates the so
called ``min-max" consensus, as a result of applying the $\min$ and
$\max$ operations in the network update step. This can be considered
as a counterpart of average consensus in set-valued settings.
\begin{figure}
  \begin{center}
      \includegraphics[width=.5\textwidth]{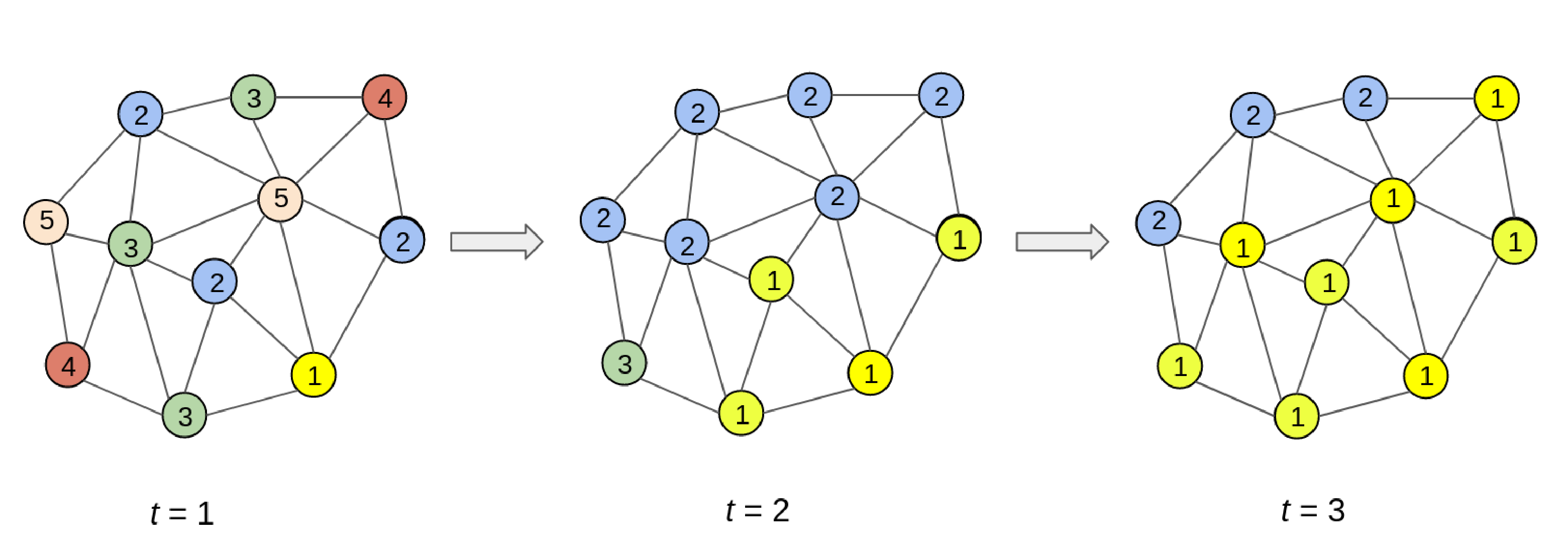}
  \end{center}
  \caption{\small{Simple static example of ``min" consensus.}}
  \label{fig:cons}
\end{figure}

We conclude this section by showing that the proposed algorithm
constructs a distributed resilient interval framer in the sense of Definition
\ref{defn:dist_obs} for the plant \eqref{eq:ind_system}.
\begin{lem}[Distributed Resilient Interval Framer Construction]
  \label{lem:ind_framer}
  Suppose that all the conditions and assumptions in
  Lemma~\ref{lem:equiv} hold. Then, Steps i) - iv) construct a
  distributed resilient interval framer for
  \eqref{eq:ind_system}.
  % \margin{After our discussions, I would say for (10)}
  % Scott: keeping as 6
\end{lem}

\begin{proof}
  % \margin{This seems to be a restatement of things that we have said
  %   and kind of proved in steps i) to iv) above. For example, we have
  %   already said prior to (16) that the system guarantees to bound the
  %   true state values. Similarly, we could say that (19) guarantees to
  %   bound the true disturbance values. Then, let's not repeat the same
  %   things here. I'm cutting this down}
  %   Scott: I'm fine with cutting it, but let's discuss
  %   Applying Propositions \ref{prop:tight_decomp} and
  %   \ref{prop:bounding} to all the uncertain terms in the right hand
  %   side of \eqref{eq:gamma_dyn_2} and \eqref{eq:d_x} shows that
  From our previous discussion on the properties of~\eqref{eq:framers}
  and~\eqref{eq:input_framers}, the following implications hold:
  \begin{align*}
    \ul{x}^{i}_{k} &\leq x_{k} \leq \ol{x}^{i}_{k}
    \implies \ul{x}^{i,0}_{k+1} \leq x_{k+1} \leq
    \ol{x}^{i,0}_{k+1}, \\
    \ul{d}^{i}_{k} &\leq d_{k} \leq \ol{d}^{i}_{k}
    \implies \ul{d}^{i,0}_{k+1} \leq d_{k+1} \leq
    \ol{d}^{i,0}_{k+1},
  \end{align*}
  for each $i\in\nodes$. % where $\ul{x}^{i,0}_{k+1},\ol{x}^{i,0}_{k+1}$ and
  % $\ul{d}^{i,0}_{k+1},\ol{d}^{i,0}_{k+1}$ are given in
  % \eqref{eq:framers} and \eqref{eq:input_framers}, respectively.
  %This means that individual framers/interval estimates are
  %correct.
  When the framer condition is satisfied for all nodes, the
  intersection of all the individual estimates of neighboring nodes
  (cf. \eqref{eq:network_update} and \eqref{eq:network_update_d}) also
  results in correct interval framers, i.e.,
  \begin{align*}
    \ul{x}^{i,0}_{k} &\leq x_{k} \leq \ol{x}^{i,0}_{k},
    \; \forall i \in \nodes
    \implies \ul{x}^{i}_{k} \leq x_{k} \leq \ol{x}^{i}_{k},
    \; \forall i \in \nodes,\\
    \ul{d}^{i,0}_{k} &\leq d_{k} \leq \ol{d}^{i,0}_{k},
    \; \forall i \in \nodes
    \implies \ul{d}^{i}_{k} \leq d_{k} \leq \ol{d}^{i}_{k},
    \; \forall i \in \nodes.
  \end{align*}
  Since the initial interval is known to all nodes~$i$, then by
  induction, Steps \hyperlink{step1}{i)}--\hyperlink{step4}{iv)}
  % \margin{again, refer to the steps of the algorithm} construct a
  % correct distributed resilient interval framer
  % Scott: resolved
  for~\eqref{eq:ind_system}.
\end{proof}
\section{Distributed Resilient Interval Observer Synthesis}
\label{sec:synthesis}
In this section, we investigate conditions on the observer gains
$L^i$, $T^i$, and $\Gamma^i$, $i \in \nodes$, as well as the
communication graph $\graph$, that lead to a C-ISS distributed
resilient interval observer (cf. Definition \ref{defn:stability}),
which equivalently results in a uniformly bounded observer error
sequence $\{e_{x,k},e_{d,k}\}_{k \ge 0}$ (given in
\eqref{eq:error_1}--\eqref{eq:error_2}), in the presence of bounded
noise.
\subsection{Stability of the Observer Design}
\noindent\textbf{\textit{Switched System Perspective.}}
Leveraging a switched system representation of the error system, we
can provide a condition that is necessary and sufficient for the
stability of the error comparison system, and, consequently, is
sufficient to guarantee the stability of the original error system. We
begin by stating a preliminary result that expresses the observer
error dynamics in the form of a specific switched system.
\begin{lem}
  \label{lem:switched}
  The collective error signals ($\{e_{x,k},e_{d,k}\}_{k=0}^{\infty}$)
  satisfy the following switched comparison dynamics:
  % \margin{I named this ``comparison dynamics''. Is that OK?}
  % Scott: yes
  \begin{align}\label{eq:error-switched}
    \begin{split}
      e_{x,k+1} &\leq \sigma^x_k(\Acompx e_{x,k} +
                  \Bcompx \delta_\eta), \\
      e_{d,k} &\leq  \sigma^d_k(\Acompd e_{x,k} + \Bcompd \delta_\eta),
    \end{split}
  \end{align}
  % \margin{could we use a product notation for the sigmas? (since they are matrices)}
  % Scott: resolved by explicitly saying they are matrices
  where $\delta_\eta \triangleq \ol{\eta} - \ul{\eta}$, see Appendix~\ref{sec:framer_mat}, for matrices
  \begin{align*}
    \sigma^x_k &\in \Sigma^x \triangleq \left\{\sigma \in \{0,1\}^{Nn \times Nn} :
    \begin{array}{c} \sigma_{ij} = 0, \forall j \notin \mathcal{N}_i, \\
    \sum_{k=1}^{Nn} \sigma_{ik} = 1 \end{array} \right\},\\
    \sigma^d_k &\in \Sigma^d \triangleq \left\{\sigma \in \{0,1\}^{Np \times Np} :
    \begin{array}{c} \sigma_{ij} = 0, \forall j \notin \mathcal{N}_i, \\
    \sum_{k=1}^{Nn} \sigma_{ik} = 1 \end{array} \right\},
  \end{align*}
  in the sets $\Sigma^x,\Sigma^d$  of possible switching signals,
  and
  \begin{gather*}
    \Acompx \triangleq \diag(\Acompx^1, \dots \Acompx^n), \quad
    \Acompd \triangleq \diag(\Acompd^1, \dots \Acompd^n), \\
    \Bcompx \triangleq \diag(\Bcompx^1, \dots \Bcompx^n), \quad
    \Bcompd \triangleq \diag(\Bcompd^1, \dots \Bcompd^n).
  \end{gather*}
  The individual matrices $\Acompx^i$,
  $\Acompd^i$, $\Bcompx^i$, and $\Bcompd^i$ are given in Appendix
  \ref{sec:error_mat}.
  % \margin{we also need to define the variables
  % $\delta_\eta$, which don't show up to this point.}
  % Scott: done
  Furthermore, $\sigma^x_k$ and $\sigma^d_k$ are binary matrices that select the neighbor with the smallest  error, i.e.,
  \begin{gather}
    \label{eq:H}
    \begin{split}
      (\sigma^x_k)_{\id(i,s),\id(j^*,s)}
      = 1
      & \Leftrightarrow j^* = \min (\argmin_{j\in\mathcal{N}_i} (e^{j,0}_k)_s),
    \end{split}
  \end{gather}
for
$e^0_{x,k} \triangleq \ol{x}^0_k - \ul{x}^0_k$,
$s \in \{1,\dots,n\}$ and $i\in\mathcal{V}$.
% \margin{I moved here this def as it appears in the statement}
% Scott: looks good
Here $\id(i,s) = n(i-1) + s$ encodes the indices associated with state
dimension $s$ at node~$i$ (and, similarly, for $\sigma^d_k$).
\end{lem}
\begin{proof}
  The proof is provided in Appendix \ref{sec:lem4_proof}.
\end{proof}
\begin{cor}
  The matrix $\sigma^x_k\Acompx$ is a member of the set
  $\mathcal{F} \subseteq \R^{Nn\times Nn}$, where
  \begin{gather*}
    \mathcal{F}
    \triangleq \Big\{
      F \in \real^{Nn \times Nn} :
      (F)_{\id(i,s)} \in \mathcal{F}^i_s,
         \ s \in \{1,\dots,n\},
         \ i\in \mathcal{V}
         \Big\} , \\
    \mathcal{F}^i_{s}
    \triangleq \Big\{
      \mathbf{e}_j^\top \otimes
        (\Acompx^j)_s
    \in \R^{1\times Nn}
      : \ {j\in\mathcal{N}_i}
      \Big\}.
  \end{gather*}
  % Scott: but we don't need to do this, since the input estimates
  % don't really have any dynamics of their own
  % \margin{Don't we need it for the next proposition? if we do, can we add a bit of detail of how to reach this conclusion?}
  % Scott: we don't need this specific property for \sigma^d A_d
  % Moreover, a similar conclusion can be derived for the the input error
  % switching dynamics $\sigma^d_k\Acompd$ with slight
  % modifications.
 % \bulletend
\end{cor}

Recall that the switching dynamics in \eqref{eq:error-switched}
depends on the state according to \eqref{eq:H} and always creates the
smallest possible error. In order to take advantage of this property
we observe that the set $\mathcal{F}$ has a specific structure known
as \emph{independent row uncertainty}, formally defined below.
\begin{defn}[Independent Row Uncertainty \cite{VDB-YN:10}]
  A set of matrices $\mathcal{M} \subset \R^{n\times n}$ has
  independent row uncertainty if
  \begin{align*}
    \mathcal{R} = \left\{
      \begin{bmatrix}a_1^\top & \cdots & a_n^\top\end{bmatrix}^\top
      \ : \ a_i \in \mathcal{R}_i, \ i \in \{1,\dots,n\}
    \right\},
  \end{align*}
  where all sets $\mathcal{R}_i \subset \R^{1\times n}$ are
  compact.
  % \margin{can we replace $\mathcal{M}_i$ by $\mathcal{R}_i$? for 'row'}
  % \margin{don't forget to add bullets}
  % Scott: resolved both
  \bulletend
\end{defn}
Next, we restate the following lemma on the spectral properties of the
sets with independent row uncertainty, that will be used later in
our stability analysis of system \eqref{eq:error-switched}.
\begin{prop}\label{lem:lower-spectral-radius}\cite[Lemma 2]{VDB-YN:10}
  Suppose $\mathcal{R} \subset \R^{n\times n}$ has independent
  row uncertainty. Then there exists
  $R_* \in \mathcal{R}$ such that:
  \begin{gather*}
    \rho(R_*) = \min_{R\in\mathcal{R}}\rho(R) =
    \lim\limits_{k\to\infty}\big{(}\min_{R_i\in\mathcal{R}} \|R_1\cdots R_k\|^\frac{1}{k}\big{)}.
\end{gather*}
The latter is known as the \emph{lower spectral radius} of
$\mathcal{R}$. \bulletend
\end{prop}
We can now state our first main stability result.
% Scott: the "consequently" part doesn't hold anymore since it's a
% comparison system
\begin{thm}[Necessary and Sufficient Conditions for Stability, Implying the C-ISS Property]
  % \margin{can we add ``implying the C-ISS property''}
  % Scott: added to the title of the theorem. Discuss?
  \label{thm:stability}
  The noiseless ($\delta_\eta=0$) comparison error system \eqref{eq:error-switched} is
  globally exponentially stable if and only if there exists
  $\sigma^x_* \in \Sigma^x$ such that the matrix
  $\sigma^x_*\Acompx$ is Schur stable. Consequently, the
  distributed observer \eqref{eq:framers}--\eqref{eq:network_update_d}
  is C-ISS if such a $\sigma^x_*$ exists.
\end{thm}
\begin{proof}
The proof is given in Appendix \ref{sec:Thm1_proof}.
\end{proof}

\subsection{C-ISS and Error Minimizing Observer Synthesis}
% \margin{I have rephrased this paragraph to emphasize the
%   contributions, which are the derivation of two tractable problem
%   reformulations which leverage theorem 1. One has to be careful with
%   providing optimizations to solve problems. First, arguably, any
%   control problem can be translated into an optimization (...) Many
%   people don't see merit in this for this reason. But reformulating
%   problems so that they are tractable or that use previous analysis
%   results makes the approach not straightforward. We have to emphasize
%   this is the case here.}
% Scott: looks good
This section contributes two different procedures for the design and
optimization of the observer gains, in order to reduce conservatism.
These methods leverage the previous characterization of
Theorem~\ref{thm:stability}, leading to a first optimization in
Lemma~\ref{lem:opt_cent}. After this, we obtain two tractable problem
reformulations: The first method, which requires central knowledge of
all system parameters, takes the form of a mixed-integer linear
program (MILP), where the number of constraints and decision variables
is of the order of $(Nn)^2$. The second and more tractable procedure
reduces this large MILP into $2N$ smaller optimization problems, which
may be solved much more efficiently at the cost of some additional
conservatism.

\subsubsection{\textbf{First Approach}}
Essentially, this approach identifies an optimization problem to
synthesize the matrix $\sigma_*^x$,
% \margin{note that in the previous
%   section we have used the notation $\sigma_*^x$ I think I like the
%   latter notation more, so please make sure this notation is
%   consistent throughout the paper}
% Scott: fixed
together with the free gains $L^i$, $T^i$, and $\Gamma^i$ introduced
in Lemma~\ref{lem:equiv}, in order to guarantee stability of the error
system via
% \margin{refer to the observer gains after lemma 1, which were left
% undefined}
% Scott: rephrased this sentence
Theorem~\ref{thm:stability}. In addition, it optimizes the performance
of the observer by minimizing the $\ell_1$-norm of the observer error
dynamics in response to the bounded noise terms.

\begin{lem}
  \label{lem:opt_cent}
  If the following optimization problem
  \begin{align} \label{eq:cent-objective}
  \begin{array}{rl}
   & \min\limits_{L, T, \Gamma, \gamma, p, \sigma}\quad
     \gamma \\
   & \mathrm{s. t.}
      \begin{bmatrix}
        p \\
        \mathbf{1}_{Nn+p}
      \end{bmatrix}^\top
      \begin{bmatrix}
        \sigma\mathcal{A}- I_{Nn} & \sigma \mathcal{B} \\ I_{Nn} & \mathbf{0} \\ \mathbf{0} & -\gamma I_{p}
      \end{bmatrix} < 0, \\
    &
     \quad \ \sigma \in \Sigma^x, \ p > 0, \quad T^i = I_n - \Gamma C_2^i, \ \forall i \in \nodes,
      \end{array}
      \end{align}
      with
      \begin{align}\label{eq:cent-constraint-last}
      \begin{array}{rl}
    & \mathcal{A} \triangleq \diag(\mathcal{A}^1, \dots, \mathcal{A}^n), \ \mathcal{B} \triangleq \diag(\mathcal{B}^1, \dots, \mathcal{B}^n), \\
    & \mathcal{A}^i \triangleq |T^i A^i - L^i C^i_2| + |T^i|\ol{F}_{\rho,x}, \ \forall i \in \nodes, \\
    %& \Bcompx = \diag(\Bcompx^1, \dots \Bcompx^n), \\
    & \mathcal{B}^i \triangleq |\Psi^i| +
      \begin{bmatrix} |T^i|\ol{F}_{\rho,w} & 0 & 0 \end{bmatrix}, \ \forall i \in \nodes,
   % & T^i = I_n - \Gamma C_2^i, \ \forall i \in \nodes,
   \end{array}
  \end{align}
  is feasible, then the comparison system \eqref{eq:error-switched} is
  C-ISS.  Furthermore, letting $\gamma^*$ be the value of the
  objective~\eqref{eq:cent-objective}, the error is upper bounded by
  the expression
\begin{gather}
  \label{eq:error-bound}
  \|e^x_k\|_1 < \gamma^* \|\delta_\eta\|_1.
\end{gather}
\end{lem}
\begin{proof}
  The proof is given in Appendix~\ref{sec:opt-cent}.
\end{proof}
% \margin{As I said before, the result would be much stronger if we can
%   provide some conditions (on the system matrices and nonlinear
%   dynamics) that lead to the feasibility of the problem. }
% Scott: this is the best we have so far

Although
the optimization problem in Lemma~\ref{lem:opt_cent} has nonlinear
constraints, it can be reformulated into an MILP by a change of
variables, which is formalized through the following theorem. Even for
large system dimensions, this MILP can be tractably solved to global
optimality by state-of-the-art solvers such as Gurobi \cite{gurobi}.
\begin{thm}
  \label{thm:milp}
  The program
  \eqref{eq:cent-objective} is
  equivalent to the MILP
  \begin{align}\label{eq:cent-objective_2}
  \begin{array}{rl}
   &\hspace{-.5cm} \min\limits_{\tilde{L}, \tilde{T}, \tilde{\Gamma}, \gamma, Q, \sigma, \mathbf{A}, \mathbf{B}}\quad
     \gamma \\
    &\quad \mathrm{s. t.}   \
        \mathbf{1}_{\tilde{n}}^\top
      \begin{bmatrix}
        \mathbf{A} - Q & \mathbf{B} \\ I_{Nn} & \mathbf{0} \\ \mathbf{0} & -\gamma I_{p}
      \end{bmatrix} < 0, \ \eqref{eq:mixed_integer} \ \text{holds,} \\
    & \quad \quad \ \
      \sigma \in \Sigma^x, \ Q > 0,  \tilde T^i = Q - \tilde\Gamma C_2^i, \ \forall i \in \nodes,
      \end{array}
      \end{align}
      where $\tilde{n} \triangleq 2Nn+p$, and
      \eqref{eq:mixed_integer} represents the additional mixed-integer
      conditions obtained using the so called ``big-$M$" approach
      \cite{JNH:12}, as follows:
   \begin{align}\label{eq:mixed_integer}
   \begin{array}{rl}
   % & \tilde T^i = Q - \tilde\Gamma C_2^i, \ \forall i \in \nodes, \\
    & -(I-\sigma_{ij})M \le \mathbf{A}_{ij} - \tilde{\mathcal{A}}^j \le (I-\sigma_{ij})M, \\
    & -\sigma_{ij} M \le \mathbf{A}_{ij} - \tilde{\mathcal{A}}^j \le \sigma_{ij} M, \\
    & -(I-\sigma_{ij})M \le \mathbf{B}_{ij} - \tilde{\mathcal{B}}^j \le (I-\sigma_{ij})M, \\
    & -\sigma_{ij} M \le \mathbf{B}_{ij} - \tilde{\mathcal{B}}^j \le \sigma_{ij} M,
      % TODO correct indices
      \end{array}
  \end{align}
  with $M \in \mathbb{R}$ chosen sufficiently large such that
  $M >\max (\max_{i,j} (\tilde{\mathcal{A}})_{ij},\max_{i,j}
  (\tilde{\mathcal{B}})_{ij})$.  Here,
  \begin{align*}
    & \tilde{\mathcal{A}} = \diag(\tilde{\mathcal{A}}^1, \dots, \tilde{\mathcal{A}}^n),
      \ \tilde{\mathcal{B}} = \diag(\tilde{\mathcal{B}}^1, \dots, \tilde{\mathcal{B}}^n), \\
    & \tilde{\mathcal{A}}^i = |\tilde{T}^i A^i - \tilde{L}^i C^i_2| + |\tilde{T}^i|\ol{F}_{\rho,x}, \ \forall i \in \nodes, \\
    % & \\
    & \tilde{\mathcal{B}}^i = |\tilde{\Psi}^i| +
      \begin{bmatrix} |\tilde{T}^i|\ol{F}_{\rho,w} & 0 & 0 \end{bmatrix}, \ \forall i \in \nodes.
  \end{align*}
  Furthermore, the optimizers in \eqref{eq:cent-objective} and \eqref{eq:cent-objective_2}
  are related as: %by the expressions
%  \begin{align*}
    $$L = Q^{-1}\tilde{L}, \ T = Q^{-1}\tilde{T}, \ \Gamma = Q^{-1}\tilde{\Gamma}.$$
%  \end{align*}
\end{thm}
\begin{proof}
  The proof is given in Appendix~\ref{sec:milp}.
\end{proof}
The optimization problem in Theorem~\ref{thm:milp} is a mixed-integer
linear program due to the linearity of the constraints in
\eqref{eq:mixed_integer} and \eqref{eq:cent-objective_2} and the fact
that the entries of the matrix $\sigma$ is restricted to take values
of either $0$ or $1$.
% \margin{Explain why this is a MILP, it may not be obvious for some
% reviewers}
% Scott: added sentence above
\subsubsection{\textbf{Second Approach}}
Alternatively to the previous centralized method, we show that the
C-ISS property implied by Theorem \ref{thm:stability} can be tractably
established in a \emph{distributed} manner.  The approach is
conceptually similar to Lemma~\ref{lem:opt_cent}, but with some
simplifying assumptions that allow for the problem to be fully
decoupled and solved in a distributed way. The design approach has two
steps: first, agents solve a linear program in order to verify an
assumption that guarantees stability of the observer. Then, using
information from the first step, they solve a second MILP in order to
minimize an upper bound on the norm of the observer errors. We begin
by describing the simplified assumption that leads to stability.

\textbf{Stabilization:} As noted above, multiplication by the matrix $\sigma^x_k$ has the
effect of permuting the rows of $\Acompx$. We now derive a sufficient
condition for stability that leverages this property.
\begin{assumption}[Collective Positive Detectability over Neighborhoods (CPDN)]
  \label{ass:col_pos_det_neigh}
  For every state dimension $s = 1,\dots,n$ and every agent
  $i\in\nodes$, there is an agent $\nu_{is}\in\mathcal{N}_i$
  such that there exist gains $T^{\nu_{is}}$, $L^{\nu_{is}}$, and
  $\Gamma^{\nu_{is}}$ satisfying
  \begin{align*}
    (\Acompx^{\nu_{is}} \mathbf{1})_s < 1.
  \end{align*}
%  \bulletend
\end{assumption}
Intuitively, the CPDN assumption narrows the problem of stability to
subgraphs. Within these subgraphs, we require that for each state
dimension $s$, there is a node that, given estimates of all other
state dimensions $\{1,\dots,s-1,s+1,\dots,n\}$, can compute an
accurate estimate of dimension~$s$. The assumption can be easily
verified by solving a linear program at every node and communicating
the results with neighbors. The purpose of the LP is to identify the
state dimensions $s$ which a node can contribute to estimating. In the
notation of Assumption~\ref{ass:col_pos_det_neigh}, each node $i$
identifies the dimensions $s$ for which it can act as $\nu_{js}$ for
$j \in \mathcal{N}_i$. The following Lemma shows that the existence
conditions in Assumption~\ref{ass:col_pos_det_neigh} can be verified
by examining the solutions of these LPs.
\begin{lem}
  For all $i \in \nodes$, let $T^i_*$, $L^i_*$, and $\Gamma^i_*$ denote the solutions to
  \begin{align}\label{eq:stab_LMI}
    \begin{array}{rl}
      &\min\limits_{\{X^i,Y^i,Z^i,L^i,T^i,\Gamma^i\}} \quad \sum_{s=1}^n \sum_{t=1}^n X^i_{st}+Y^i_{st} \\
      &\quad \quad  \quad  \mathrm{s.t.}
        \begin{cases}
          -X^i \leq T^i{A}^i-L^iC_2^i \leq X^i,\\
          -Z^i \leq T^i \leq Z^i,\\
          0 \leq Z^i\ol{F}_{\rho,x} \leq Y^i,\\
          T^i = I_n - \Gamma^iC_2^i.
        \end{cases}
    \end{array}
  \end{align}
  Then Assumption~\ref{ass:col_pos_det_neigh} holds if and only if for all $i \in \nodes$ and $s \in \{1,\dots,n\}$, there is a $\nu_{is} \in \mathcal{N}_i$ such that
  \begin{align}\label{eq:lp_condition}
    ((|T^{\nu_{is}}_* A^{\nu_{is}} - L^{\nu_{is}}_* C^{\nu_{is}}| + |T^{\nu_{is}}_*|\ol{F}_{\rho,x}) \mathbf{1})_s < 1.
  \end{align}
%  \bulletend
\end{lem}
Since the condition in~\eqref{eq:lp_condition} must hold at every node, it can be verified in a distributed manner. The entire verification procedure is summarized in
Algorithm~\ref{alg:init}. If the condition in Line~\ref{line:cpdn} is
false for any $i$, it implies that
Assumption~\ref{ass:col_pos_det_neigh} is not satisfied, so the
algorithm returns false. Otherwise, the algorithm returns the set
$\mathbb{J}_i$, which will be used to further optimize the observer
gains.

\begin{algorithm}
  \setstretch{1.1}
  \caption{CPDN verification at node $i$.}
  \label{alg:init}
  \begin{algorithmic}[1]
    \renewcommand{\algorithmicrequire}{\textbf{Input:}}
    \renewcommand{\algorithmicensure}{\textbf{Output:}}
    \Require $A$, $C^i$, $\mathcal{N}_i$;
    \textbf{Output:} $\mathbb{J}^i$
    \State Compute $L_*^{i}$, $\Gamma_*^{i}$, and $Z_*^{i}$ by solving the LP in \eqref{eq:stab_LMI}.
    \State $\mathbb{J}^i \gets \{s \ : \ \sum_{t=1}^n(Z_*^i)_{st} < 1\}$;
    \State $\mathcal{Q}_i \gets \{(I - \Gamma^i_*C_2^i) A^i - L_*^{i}C_2^{i}\}$;
    \State Receive $\mathcal{Q}_j$ from $j\in\mathcal{N}_i$;
    \State $\mathcal{Q}_i \gets \bigcup_{j\in\mathcal{N}_i} Q_j$;
    \If{$\forall s \in \{1,\dots,n\}$, $\exists P \in \mathcal{Q}_i$ s.t. $\|(P)_s\|_1 < 1$}\label{line:cpdn}
    \Return $\mathbb{J}^i$
    \Else \
    \Return false (i.e., Assumption~\ref{ass:col_pos_det_neigh} not satisfied)
    \EndIf
  \end{algorithmic}
\end{algorithm}

The following theorem formalizes the importance of the LP
\eqref{eq:stab_LMI} in designing a stable observer.
% \margin{As before: comments on the feasibility of this assumption?
% maybe we can refer the reader to later to discuss this}
% Scott: solved with reordering
\begin{thm}
  \label{thm:cpdn_verif}
  Suppose Assumptions
  \ref{assumption:mix-lip}--\ref{ass:col_pos_det_neigh} hold. Then,
  the proposed distributed observer
  \eqref{eq:framers}--\eqref{eq:network_update_d} is C-ISS with the
  corresponding observer gains ${L}^{*,i}$, ${T}^{*,i}$, and
  ${\Gamma}^{*,i}$ that are solutions to \eqref{eq:stab_LMI}
\end{thm}
\begin{proof} The proof is provided in Appendix \ref{sec:peas_proof}.
\end{proof}

\textbf{Error-Minimization:} After computing the sets
$\mathbb{J}_i$, each node can further optimize its gains to reduce the
overall observer error while maintaining the stability guarantees from
the previous section. Each node solves the MILP in
\eqref{eq:stab_min_LMI}, which as will be shown in Theorem
\ref{thm:suff_stability}, simultaneously \emph{guarantees stability}
and minimizes an upper bound on the observer error. In this way, the
design includes a sense of noise/error attenuation. To this end, we
first provide a preliminary result on how to calculate the proposed
observer steady state errors.
\begin{lem}[Error Bounds]
  \label{lem:errors}
  Suppose all the Assumptions in Theorem \ref{thm:cpdn_verif} hold and
  consider the proposed distributed observer in
  \eqref{eq:framers}--\eqref{eq:network_update_d} , where the observer
  gains $T^i,L^i,\Gamma^i$ are solutions to the LP in
  \eqref{eq:stab_LMI}. Then, for all
  $\sigma^x\in \Sigma^x,\sigma^d\in \Sigma^d$, the
  observer error sequences are upper bounded as follows:
  \begin{align}\label{eq:error_upper_bounds}
    \begin{split}
      \|e_{x,k}\|_{\infty} &\leq \|e_{x,0}\|_{\infty}\rho_*^k
                             + \frac{1-\rho_*^k}{1-\rho_*}
                             \max_{i}
                             \|\pi^i_x\|_\infty, \\
      \|e_{d,k}\|_\infty &\leq \rho(\Acompd)\|e_{x,k}\|_\infty
                           + \max_i \|\pi^i_d\|_\infty,
    \end{split}
  \end{align}
  where
 % \begin{align*}
    $\rho_* \triangleq \rho(\sigma\Acompx)$,
  %\end{align*}
  and $\pi^i_x,\pi^i_d$ are given in Appendix \ref{sec:bound_mat}.
\end{lem}
\begin{proof}
  The proof is given in Appendix \ref{sec:error_bounds_proof}.
\end{proof}
Now, equipped with the results in Lemma \ref{lem:errors}, we are ready
to formalize our next main results on how to tractably synthesize
stabilizing and error minimizing observer gains in a distributed
manner.
\begin{thm}[Distributed Optimal Gain Design]
  \label{thm:suff_stability}
  Suppose Assumptions
  \ref{assumption:mix-lip}--\ref{ass:col_pos_det_neigh} hold and
  $L^i_*$, $T^i_*$, and $\Gamma^i_*$ are solutions to the following
  MILP:
  \begin{align}\label{eq:stab_min_LMI}
    \begin{array}{rl}
      &\min\limits_{\{X^i,Y^i,Z^i,L^i,T^i,\Gamma^i\}}
        \|\Pi^i_w\delta_w+\Pi^i_v\delta^i_v\|_\infty \\
      & \mathrm{s.t.}
        \begin{cases}
          \sum_{t=1}^n X^i_{jt}+Y^i_{jt} < 1, \ \forall j \in \mathbb{J}^i,\\
          -X^i \leq T^i{A}^i-L^iC_2^i \leq X^i,\\
          -Z^i \leq T^i \leq Z^i,\\
          0 \leq Z^i\ol{F}_{\rho,x} \leq Y^i,\\
          T^i = I_n - \Gamma^iC_2^i,
        \end{cases}
    \end{array}
  \end{align}
  where
  \begin{align*}
    \Pi^i_w &\triangleq |T^iB^i|+|T^i|\ol{F}_{\rho,w},\\
    \Pi^i_v &\triangleq |T^i\Phi^iD^i_1+L^iD^i_2|+|(T^iG^i_2M^i_2+\Gamma^i)D^i_2|,
  \end{align*}
  and $\mathbb{J}^i$ is calculated using Algorithm \ref{alg:init}.  Then,
  the {\dsiso} algorithm, i.e., the proposed distributed recursive
  algorithm in \eqref{eq:framers}--\eqref{eq:network_update_d}, with
  the corresponding observer gains $L^i_*, T^i_*, \Gamma^i_*$
  constructs a C-ISS distributed input and state interval observer.

  Moreover, the steady state observer errors are guaranteed to be
  bounded as
  follows:
  \begin{align}\label{eq:ss_upperbounds}
    \begin{split}
      \|e_{x,k}\|_{\infty} &\leq \frac{1}{1-\rho_*}
                             \max_{i}
                             \|\pi^i_x\|_\infty, \\
      \|e_{d,k}\|_\infty &\leq \frac{\rho(\Acompd)}{1-\rho_*} \max_{i}
                           \|\pi^i_x\|_\infty + \max_i \|\pi^i_d\|_\infty,
    \end{split}
  \end{align}
  with $\rho_*$, $\Acompd,\pi^i_x$, and $\pi^i_d$ given in
  Lemma \ref{lem:errors} and Appendix \ref{sec:error_bounds_proof}.
\end{thm}
\begin{proof}
  The proof can be found in Appendix \ref{sec:error_min_proof}.
\end{proof}
We conclude this section with Algorithm~\ref{alg1}, which summarizes
the proposed distributed simultaneous input and state observer
(\dsiso), whose operation is the same regardless of which gain design
method is used.
\begin{algorithm}[]
  \setstretch{1.1}
  \caption{\dsiso \ at node $i$.}
  \label{alg1}
  \begin{algorithmic}[1]
    \renewcommand{\algorithmicrequire}{\textbf{Input:}}
    \renewcommand{\algorithmicensure}{\textbf{Output:}}
    \Require %\textbf{Input:}
    $\ul{x}^i_0$, $\ol{x}^i_0$,; \textbf{Output:}
    $\{\ul{x}^i_{k},\ol{x}^i_{k},\ul{d}^i_{k},\ol{d}^i_{k}\}_{k\ge0}$;
    \State Compute $L^i$ $\Gamma^i$,
    and $T^i$ by solving \eqref{eq:cent-objective_2} or \eqref{eq:stab_min_LMI};
    \State $k \gets 1$
    \Loop
      \Statex {$\triangleright$ \textbf{State propagation and measurement update}}
      \State Compute $\ul{x}^{i,0}_{k}$ and
      $\ol{x}^{i,0}_{k}$ using \eqref{eq:framers};
      \Statex {$\triangleright$ \textbf{State framer network update}}
      \State Send $\ul{x}^{i,0}_{k}$ and
      $\ol{x}^{i,0}_{k}$ to $\{j \ : \ i \in \mathcal{N}_j\}$;
      \State Receive $\ul{x}^{j,0}_{k}$ and
      $\ol{x}^{j,0}_{k}$ from $j \in \mathcal{N}_i$;
      \State $\displaystyle\ul{x}^{i}_{k} \gets\,
      \max_{j\in\mathcal{N}_i}\ul{x}^{j,0}_{k}; \quad
      \ol{x}^{i}_{k} \gets\, \min_{j\in\mathcal{N}_i}\ol{x}^{j,0}_{k};$
      \Statex {$\triangleright$ \textbf{Input framer estimation}}
      \State Compute $\ul{d}^{i,0}_{k}$ and
      $\ol{d}^{i,0}_{k}$ using \eqref{eq:input_framers};
      \Statex {$\triangleright$ \textbf{Input framer network update}}
      \State Send $\ul{d}^{i,0}_{k}$ and
      $\ol{d}^{i,0}_{k}$ to $\{j \ : \ i \in \mathcal{N}_j\}$;
      \State Receive $\ul{d}^{j,0}_{k}$ and
      $\ol{d}^{j,0}_{k}$ from $j \in \mathcal{N}_i$;
      \State     $\displaystyle\ul{d}^{i}_{k} \gets\,
      \max_{j\in\mathcal{N}_i}\ul{d}^{j,0}_{k}; \quad
      \ol{d}^{i}_{k} \gets\, \min_{j\in\mathcal{N}_i}\ol{d}^{j,0}_{k};$
      \State $k \gets k+1$;
    \EndLoop

   \noindent \Return $\{\ul{x}^i_{k},\ol{x}^i_{k},\ul{d}^i_{k},\ol{d}^i_{k}\}_{k\ge0}$
  \end{algorithmic}
\end{algorithm}
%\section{Extensions \& Relaxations}
%In this section, we briefly discuss possible extensions and
%relaxations that can be easily deployed to make our distributed
%observer design applicable to even a broader class of networked
 %dynamical systems.

\section{Nonlinear Observations \& Nonlinear Attacks}
It is noteworthy that System \eqref{eq:ind_system} can be easily
extended in several ways to cover much more general classes of
nonlinear dynamics, e.g., to include the case where \emph{different}
attack signals are injected onto the sensors and actuators as well as
the case where the attack signals compromise the system in a
\emph{nonlinear} manner. To illustrate this, consider the following
dynamical system:
\begin{align}\label{eq:ind_system_extension}
  \begin{split}
    x_{k+1}&=f(x_k,w_k)+\hat{G}{g}_k(x_k,d^s_k),\\
    y^i_k&=\sigma^i(x_k,v^i_k)+\hat{H}^i\chi_k(x_k,d^o_k), \ i \in \mathcal{V}, \ k \in \mathbb{Z}_{\geq 0},
  \end{split}
\end{align}
which is an extension of System \eqref{eq:ind_system},
$d^s_k \in \R^{p_s}$ and $d^o_k \in \R^{p_o}$ can be
interpreted as arbitrary (and \emph{different}) unknown inputs that
affect the state and observation equations through the known
\emph{nonlinear time-varying} vector fields
${g}_k:\R^n \times \R^{p_s} \to
\R^{n_{\hat{G}}}$ and
${\chi}_k:\R^n \times \R^{p_o} \to
\R^{n_{\hat{H}}}$, respectively. Moreover,
$\hat{G} \in \R^{n \times n_{\hat{G}}}$ and
$\hat{H}^i \in \R^{l^i \times n_{\hat{H}}}$ are known
time-invariant matrices.

On the other hand,
$\sigma^i:\R^n \times \R^{n^v_i} \to \R^l$ is
a known observation mapping for which we
consider two cases. \\
\emph{Case 1.} $\sigma^i(x,v)=Cx+{D}^iv$, i.e., $\sigma^i$ is linear in $x$ and $v$. \\
\emph{Case 2.} $\sigma^i$ is nonlinear with bounded interval domains,
i.e., there exist known intervals $\ol{\mathcal{X}}$ and
$\ol{\mathcal{V}}^i$ such that
$\mathcal{X} \subseteq \ol{\mathcal{X}} \subset \R^n$
and
$\mathcal{V}^i \subseteq \ol{\mathcal{V}}^i \subset
\R^{n^i_ v}$.

In the second case, we can apply our previously developed \emph{affine
  over-approximation (abstraction)} tools in reference
\cite{KRS-QS-SZY:18} to derive \emph{affine upper and lower
  over-approximations} for $\sigma^i$, using \cite[Theorem
1]{KRS-QS-SZY:18} and the linear program therein to obtain
$\ol{C}^i,\ul{C}^i,\ol{D}^i,\ul{D}^i,\ol{e}^i$ and $\ul{e}^i$ with
appropriate dimensions, such that for all $x_k \in \mathcal{X}$ and
${v}^i_k \in \mathcal{V}^i$:
\begin{align} \label{eq:abstraction1}
\ul{C}^ix_k +\ul{D}^i{v}^i_k+\ul{e}^i \leq \sigma^i(x_k,{v}^i_k) \leq \ol{C}^ix_k +\ol{D}^i{v}^i_k+\ol{e}^i,
\end{align}
Next, by taking the average of the upper and lower affine
approximations in \eqref{eq:abstraction1} and adding an additional
bounded disturbance/perturbation term $v^{a,i}_k$ (with its
$\infty$-norm being less than half of the maximum distance), it is
straightforward to reformulate the inequalities in
\eqref{eq:abstraction1} as the following equality:
\begin{align}
  \label{eq:ffine_observation}
  \sigma^i(x_k,{v}^i_k)= {C}^ix_k +{D}^i\hat{v}^i_k+{e}+v^{a,i}_k,
\end{align}
with $C^i \triangleq \frac{1}{2}(\ol{C}^i+\ul{C}^i)$,
${D}^i \triangleq \frac{1}{2}(\ol{D}^i+\ul{D}^i)$,
$e^i \triangleq \frac{1}{2}(\ol{e}^i+\ul{e}^i)$,
$\|v^{a,i}_k\|_{\infty} \leq \eta^i_{v^a} \triangleq \frac{1}{2}
\theta^i_*$, where $\theta^i_*$ is the solution to the LP in
\cite[Equation (16)]{KRS-QS-SZY:18}. In other words, the equality in
\eqref{eq:ffine_observation} is a ``redefinition" of the inequalities
in \eqref{eq:abstraction1}, which is obtained by adding the uncertain
noise $v^{a,i}_k$ to the midpoint (center) of the interval in
\eqref{eq:abstraction1} (i.e.,
${C}^ix_k +{D}^i{v}^i_k+{e}^i=\frac{1}{2}(\ul{C}^ix_k
+\ul{D}^i{v}^i_k+\ul{e}^i+\ol{C}^ix_k
+\ol{D}^i{v}^i_k+\ol{e}^i)$, to recover all possible
$\sigma^i(x_k,{v}^i_k)$ in the interval given by
\eqref{eq:ffine_observation}. In a nutshell, the above procedure
``approximates'' $ \sigma^i(x_k,{v}^i_k)$ with an appropriate linear
term and accounts for the ``approximation error" using an additional
disturbance/noise term.

Then, using \eqref{eq:ffine_observation}, the system in
\eqref{eq:ind_system_extension} can be rewritten as:
\begin{align}
  \label{eq:system2}
  \begin{array}{rl}
x_{k+1}&=f(x_k,w_k)+\hat{G}{g}_k(x_k,d^s_k), \\
    y^i_k&=C^ix_k+{D}^i_k{v}^i_k+\hat{H}^i{\chi}_k(x_k,d^o_k),i \in \mathcal{V}, k \in \mathbb{Z}_{\geq 0}. \end{array}
\end{align}
Now, courtesy of the fact that the unknown input signals $d^s_k$ and
$d^o_k$ in \eqref{eq:system2} can be completely arbitrary, by lumping
the nonlinear functions with the unknown inputs in \eqref{eq:system2}
into a newly defined unknown input signal
$d_k \triangleq \begin{bmatrix} g_k(x_k,d^s_k) \\
                  \chi_k(x_k,d^o_k)\end{bmatrix} \in \R^p$, as well as
  defining
  ${G}\triangleq \begin{bmatrix} \hat{G} & 0_{n \times
                                           n_{\hat{H}}}\end{bmatrix}$,
  ${H}^i\triangleq \begin{bmatrix} 0_{l^i \times n_{\hat{G}}} &
                                                                \hat{H}^i \end{bmatrix}$,
  we can equivalently transform system \eqref{eq:system2} to a new
  representation, precisely in the form of \eqref{eq:ind_system}.
\begin{rem}
  \label{rem1}

  From the discussion above, we can conclude that set-valued state and
  input observer designs for System \eqref{eq:ind_system} are also
  applicable to system \eqref{eq:ind_system_extension}, with the
  slight difference in input estimates that the latter returns
  set-valued estimates for
  ${d}_k \triangleq \begin{bmatrix} g_k(x_k,d^s_k) \\
                      \chi_k(x_k,d^o_k)\end{bmatrix}$, where we can
  apply any \emph{pre-image} set computation techniques in the
  literature such as reference \cite{CHN-FJW:98} to find set estimates
  for $d^s_k$ and $d^o_k$ using the set-valued estimate for $x_k$.
  \bulletend
\end{rem}
\begin{rem}
  Note that the case where the feedthrough matrix in System
  \eqref{eq:ind_system} is zero, i.e., $H^i=0$, as well as the case
  where the process and sensors in \eqref{eq:ind_system} are degraded
  by different attack (unknown input) signals, are both special cases
  of the system \eqref{eq:ind_system_extension}, where $\hat{H}^i=0$,
  ${g}_k,\chi_k$ are affine functions, respectively; thus, these cases
  can also be considered with our proposed framework.
  \bulletend
\end{rem}

\section{Illustrative Examples and Comparisons}
\subsection{Unicycle Target}
This scenario consists of a single target, modeled by a unicycle
dynamics which is controlled by an unknown agent. It is being tracked
by a network of $N=6$ agents with access to various measurements of
the position, bearing, and velocity. The goal of the agents is to
maintain consistent estimates of the target state and the unknown
control inputs. More concretely, the target has a state
$x \in \real^4$, representing the $(x,y)$ position, attitude, and
forward velocity, respectively. The state obeys the (discretized)
dynamics
\begin{align}
  \label{eq:uni-state}
  x_{k+1} = x_k + \Delta_t
  \begin{bmatrix}
    x_{4,k}\cos(x_{3,k}) + w_{1,k} \\
    x_{4,k}\sin(x_{3,k}) + w_{2,k} \\
    d_{1,k} \\ d_{2,k}
  \end{bmatrix}
\end{align}
with a time step of
$\Delta_t = 0.01$s. After performing the JSS decomposition, we
arrive at the values
\begin{gather*}
  A =
  \begin{bmatrix}
    1 & 0 & 0.01 & 0.01 \\
    0 & 1 & 0.01 & 0.01 \\
    1 & 0 & 1 & 0 \\
    1 & 0 & 0 & 1
  \end{bmatrix}, \
  B =
  \begin{bmatrix}
    0 & 0 \\
    0 & 0 \\
    0 & 0.01 \\
    0.01 & 0
  \end{bmatrix} \\
  \phi = 0.01
  \begin{bmatrix}
    x_4 \cos(x_3) - x_3 - x_4 \\
    x_4 \sin(x_3) - x_3 - x_4 \\
    0 \\
    0
  \end{bmatrix}
\end{gather*}

The communication network has a graph $\graph$ with Laplacian
\begin{gather*}
  \mathcal{L} = \begin{bmatrix}
    -3  &   1  &   0  &   1  &   1  &  0 \\
    1   & -2   &  1   &  0   &  0   &  0 \\
    0   &  1   & -3   &  1   &  0   &  1 \\
    1   &  0   &  0   & -3   &  1   &  1 \\
    0   &  1   &  0   &  1   & -2   &  0 \\
    0   &  1   &  0   &  1   &  1   & -3
  \end{bmatrix}.
\end{gather*}
Each agent has access to $l_i = 4$ measurements with randomly
generated $C^i$ matrices. The measurement noise bounds are uniformly
randomly generated on the interval $[0, 0.02]$. The measurements are
rounded to the second digit, representing an quantization error that
introduces an additional $\pm 0.005$ of measurement noise. Finally,
there is an additive noise $w_k \in \real^2$ that affects $x_1$ and $x_2$.
It satisfies $w_k \in [-10,10]\times[-10,10]$ and is used to model
slipping and random perturbations from the environment.

We design the gain matrices $L^i$ and $\Gamma^i$ using the MILP
defined in Theorem~\ref{thm:milp}. Since this MILP is feasible, we can
guarantee the observer estimates will remain bounded, as shown in
Lemma~\ref{lem:opt_cent}. The solution takes 30 seconds.

Figure~\ref{fig:uni-x} shows the resulting state framers from every
agent in the network. All agents are able to maintain a tight estimate
of the target states, with close agreement. Evidently some agents are
able to obtain slightly better estimates due to the variation in
measurement noise.
\begin{figure}[]
  \centering
  \includegraphics[width=\columnwidth]{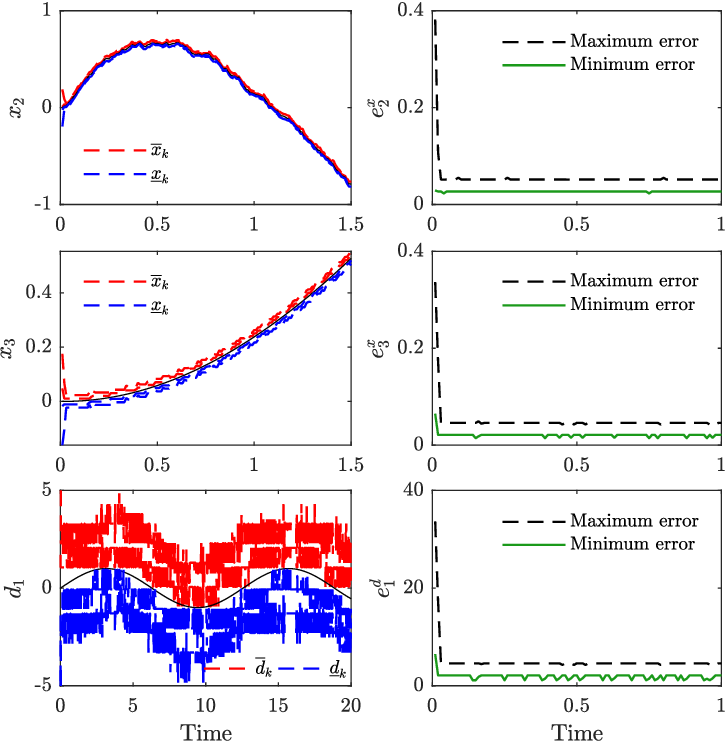}
  \caption{Framers and estimation errors for $x_2$ and $x_3$. The
    framer plots (left column) show the estimates from the worst and
    best performing agents, with upper bounds in red and lower bounds
    in blue. The error plots (right column) show the error from the
    worst performing agent in a black dashed line and the best
    performing agent in solid green.}
  \label{fig:uni-x}
\end{figure}

We conclude this example by comparing our observer with a recent
linear distributed interval observer \cite{XW-HS-FZ-GC:23} on the task
of estimating the attitude and angular velocity of the unicycle
target. For our observer, using the model described above, this means
estimating $\theta$ and $d_1$. Because the observer in
\cite{XW-HS-FZ-GC:23} is designed for linear systems and does not
handle unknown inputs, we cannot use the full unicycle model. Instead,
we adapt the attitude model, observer, and gains reported in
\cite[Section V]{XW-HS-FZ-GC:23} to apply to a single target. The
resulting model is
\begin{gather}
  \begin{bmatrix}\dot\theta \\ \ddot\theta\end{bmatrix} =
  \begin{bmatrix}0 & 1 \\ 0 & 0 \end{bmatrix}
  \begin{bmatrix}\theta \\ \dot\theta \end{bmatrix} +
  \begin{bmatrix} 0 \\ \phi \end{bmatrix}.
\end{gather}
Using this model, the goal is to estimate both $\theta$ and
$\dot\theta$, which is equivalent to estimating $x_3$ and $d_1$ in
\eqref{eq:uni-state}. Each agent has access to $y_i = \theta + \psi$, with the
same noise model as described above, and the communication graph
remains the same. It is important to note that this approach requires
an estimated bound on $\phi$, meaning a bound on the derivative of the
unknown input $d_1$. This information is not required in our
method. We use a conservative bound of $\phi \in [-2,2]$.

\begin{figure}[]
  \centering
  \includegraphics[width=\columnwidth]{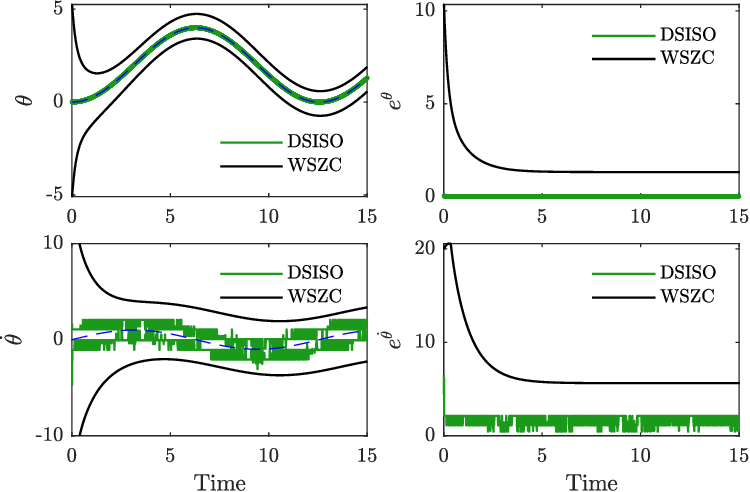}
  \caption{Framers and estimation errors for $\theta$ and
    $\dot\theta$, comparing our approach ($\dsiso$) with the observer from
    \cite{XW-HS-FZ-GC:23} (labeled WSZC).}
  \label{fig:uni-comp}
\end{figure}

Figure~\ref{fig:uni-comp} shows the results of the observer from
\cite{XW-HS-FZ-GC:23}.  Our method quickly obtains a much tighter
interval estimate of $\theta$ and attains much better estimation
performance on $\dot\theta$ ($d_1$), despite not knowing any prior
bounds. This difference is presumably due to the fact that our method
is able to incorporate the full nonlinear model into the observer,
rather than relying on the simplified linear model for the attitude
dynamics. It also highlights the importance of our gain design
procedure, which minimizes the resulting interval width.

\subsection{Power System}
\label{sec:power}
In this scenario we demonstrate the {\dsiso} algorithm on IEEE
145-bus, 50 generator dynamic test case \cite{UW:93}. We use the
effective network (EN) model \cite{TN-AEM:15} to model the dynamics of
the generators. A description of the model is beyond the scope of this
paper; for the specific parameters and equations used in our
simulation we refer the reader to reference \cite{TN-AEM:15} and the
MATLAB toolbox mentioned therein. The resulting continuous-time model
is discretized using the explicit midpoint method, to obtain equations
of the form \eqref{eq:ind_system}.  The $n=100$ dimensional state
$x_k^\top = \begin{bmatrix}\delta_k^\top &
  \omega_k^\top\end{bmatrix}^\top$ represents the rotor angle and
frequency of each of the 50 generators. Each bus in the test case
corresponds to a node in the algorithm, and we assume that the
communication network has the same topology as the power network. The
noise signals satisfy $\|w_k\|_\infty < 5$ and
$\|v^i_k\| < 1\times 10^{-4}$ $\forall i \in \nodes$.  Similarly to
the example in \cite{FP-FD-FB:12a}, each node measures its own real
power injection/consumption, the real power flow across all branches
connected to the node, and for generating nodes, the rotor angle of
the associated generator.

In this example, we assume that the generator at bus 60 is insecure
and potentially subject to attacks affecting the generator
frequency. Due to the reduction that takes place in the EN model
\cite{TN-AEM:15}, the disturbance appears additively in the
representative dynamics of all nodes, resulting in a $G$ matrix with
all non-zero entries. Due to the large system dimension and large
number of nodes, solving the MILP described in Theorem~\ref{thm:milp}
is intractable. Instead, we use Algorithm~\ref{alg:init} to verify
that Assumption~\ref{ass:col_pos_det_neigh} holds and compute
stabilizing but suboptimal observer gains. The computation takes an
average of $1.7 \pm 0.4$ (standard deviation) seconds per agent.

\begin{figure}[]
  \centering
  \includegraphics[width=\columnwidth]{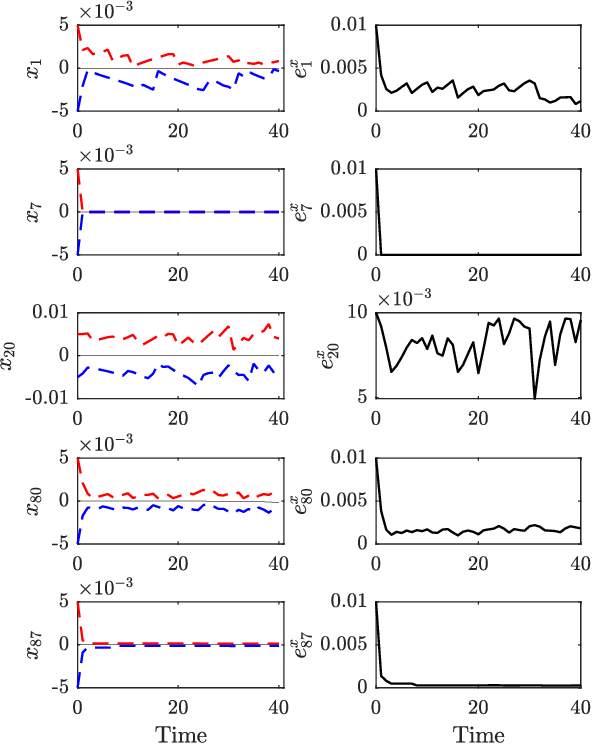}
  \caption{State framers (upper bound in red, lower bound in blue),
    \mo{as well as errors} for selected state dimensions for the power
    system example. Only the minimum error is plotted.}
  \label{fig:power_x}
\end{figure}
\begin{figure}[]
  \centering
  \includegraphics[width=\columnwidth]{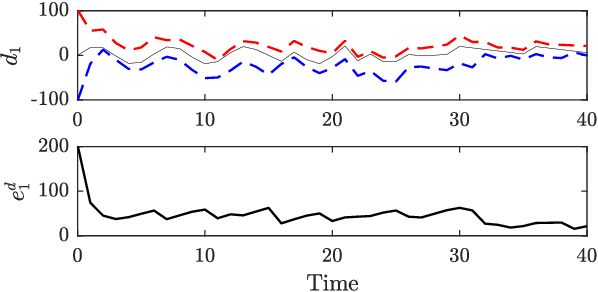}
  \caption{Input framer and framer error the frequency disturbance at
    generator 60. Only the minimum error is plotted.}
  \label{fig:power_d}
\end{figure}
Figures \ref{fig:power_x} and \ref{fig:power_d} show the input and state framers
for selected dimensions, respectively. It is clear that the algorithm
is able to estimate the state $x_1$ despite the disturbance with only
minor performance degradation. The switching due to
\eqref{eq:network_update}, which depends on the noise, is also
evident. The estimation performance for the other states is
comparatively better, since they are only affected by (known) bounded
noise. Further, all agents can maintain an accurate
estimate of the disturbance.

\section{Conclusion and Future Work}
A novel recursive distributed algorithm comprising four steps was introduced in this paper, with the objective of synthesizing input and state interval observers for nonlinear bounded-error discrete-time multi-agent systems. The systems under consideration were equipped with sensors and actuators that were susceptible to adversarial unknown disturbance signals, for which no information regarding their bounds, energy, distribution, etc., was available. The interval-valued estimates computed were ensured to encompass the true value of the states and unknown inputs. Furthermore, verifiable conditions for the stability of the proposed observer were established through two alternative approaches, both of which were shown to minimize a calculated upper bound for the interval widths of observer errors. The observer design was characterized as tractable and computationally efficient, rendering it a valuable approach to address these challenging estimation scenarios. This was demonstrated through simulations and comparisons with some benchmark observers.

Future work considers other types of adversarial signals such as communication and linkage attacks and eavesdropping malicious agents, as well as (partially) unknown dynamics.

\bibliographystyle{plain}
\bibliography{alias,SM,SMD-add}

\begin{thebibliography}{10}

\bibitem{MA-MR:18}
M.~Abolhasani and M.~Rahmani.
\newblock Robust deterministic least-squares filtering for uncertain
  time-varying nonlinear systems with unknown inputs.
\newblock {\em Systems and Control Letters}, 122:1--11, 2018.

\bibitem{AEA-AYK-FG:12}
A.~E. Ashari, A.~Y. Kibangou, and F.~Garin.
\newblock Distributed input and state estimation for linear discrete-time
  systems.
\newblock In {\em {IEEE} Int. Conf. on Decision and Control}, pages 782--787,
  2012.

\bibitem{FB-MS:12}
F.~Blanchini and M.~Sznaier.
\newblock A convex optimization approach to synthesizing bounded complexity
  $\ell^{\infty}$ filters.
\newblock {\em IEEE Transactions on Automatic Control}, 57(1):216--221, 2012.

\bibitem{VDB-YN:10}
V.~D. Blondel and Y.~Nesterov.
\newblock Polynomial-time computation of the joint spectral radius for some
  sets of nonnegative matrices.
\newblock {\em SIAM Journal on Matrix Analysis and Applications},
  31(3):865--876, 2010.

\bibitem{GC-YZ-SG-WH:21}
G.~Chen, Y.~Zhang, S.~Gu, and W.~Hu.
\newblock Resilient state estimation and control of cyber-physical systems
  against false data injection attacks on both actuator and sensors.
\newblock {\em IEEE Transactions on Control of Network Systems}, 9(1):500--510,
  2021.

\bibitem{XC-JL-PL-ZS:13}
X.~Chen, J.~Lam, P.~Li, and Z.~Shu.
\newblock $\ell_1$-induced norm and controller synthesis of positive systems.
\newblock {\em Automatica}, 49(5):1377--1385, 2013.

\bibitem{BC-JZ-GPH-SK-AWC:18}
B.~Cheng, J.~Zhang, G.~P. Hancke, S.~Karnouskos, and A.~W. Colombo.
\newblock Industrial cyber-physical systems: Realizing cloud-based big data
  infrastructures.
\newblock {\em IEEE Industrial Electronics Magazine}, 12(1):25--35, 2018.

\bibitem{MLC-AC:17}
M.~L. Corradini and A.~Cristofaro.
\newblock Robust detection and reconstruction of state and sensor attacks for
  cyber-physical systems using sliding modes.
\newblock {\em IET Control Theory \& Applications}, 11(11):1756--1766, 2017.

\bibitem{KD-XR-ASL-DEQ-LS:20}
K.~Ding, X.~Ren, A.~S. Leong, D.~E. Quevedo, and L.~Shi.
\newblock Remote state estimation in the presence of an active eavesdropper.
\newblock {\em IEEE Transactions on Automatic Control}, 66(1):229--244, 2020.

\bibitem{DE-TR-SC-AZ:13}
D.~Efimov, T.~Ra{\"\i}ssi, S.~Chebotarev, and A.~Zolghadri.
\newblock Interval state observer for nonlinear time varying systems.
\newblock {\em Automatica}, 49(1):200--205, 2013.

\bibitem{NE-DGD-DH:19}
N.~Ellero, D.~Gucik-Derigny, and D.~Henry.
\newblock An unknown input interval observer for {LPV} systems under
  ${L}_2$-gain and ${L}_{\infty}$-gain criteria.
\newblock {\em Automatica}, 103:294--301, 2019.

\bibitem{gurobi}
LLC {Gurobi Optimization}.
\newblock Gurobi optimizer reference manual, 2018.

\bibitem{JNH:12}
J.~N. Hooker.
\newblock {\em Integrated Methods for Optimization}.
\newblock Springer, 2012.

\bibitem{MK-SB-SM:22a-acc}
M.~Khajenejad, S.~Brown, and S.~Mart{\'\i}nez.
\newblock Distributed interval observers for bounded-error {LTI} systems.
\newblock In {\em {A}merican {C}ontrol {C}onference}, San Diego, CA, USA, June
  2022.

\bibitem{MK-SB-SM:22b-acc}
M.~Khajenejad, S.~Brown, and S.~Mart{\'\i}nez.
\newblock Distributed resilient interval observers for bounded-error {LTI}
  systems subject to false data injection attacks.
\newblock In {\em {A}merican {C}ontrol {C}onference}, San Diego, CA, USA, June
  2022.

\bibitem{khajenejad2023resilientcdc}
M.~Khajenejad, Z.~Jin, T.N. Dinh, and S.Z. Yong.
\newblock Resilient state estimation for nonlinear discrete-time systems via
  input and state interval observer synthesis.
\newblock In {\em 2023 62nd IEEE Conference on Decision and Control (CDC)},
  pages 1826--1832. IEEE, 2023.

\bibitem{khajenejad2021intervalmodel}
M.~Khajenejad, Z.~Jin, and S.Z. Yong.
\newblock Interval observers for simultaneous state and model estimation of
  partially known nonlinear systems.
\newblock In {\em 2021 American Control Conference (ACC)}, pages 2848--2854.
  IEEE, 2021.

\bibitem{MK-FS-SZY:22a}
M.~Khajenejad, F.~Shoaib, and S.~Z. Yong.
\newblock Interval observer synthesis for locally {L}ipschitz nonlinear
  dynamical systems via mixed-monotone decompositions.
\newblock In {\em {A}merican {C}ontrol {C}onference}, pages 2970--2975, 2022.

\bibitem{MK-SZY:19a}
M.~Khajenejad and S.~Z. Yong.
\newblock Simultaneous input and state set-valued
  $\mathcal{H}_{\infty}$-observers for linear parameter-varying systems.
\newblock In {\em {A}merican {C}ontrol {C}onference}, pages 4521--4526, 2019.

\bibitem{MK-SZY:22}
M.~Khajenejad and S.~Z. Yong.
\newblock {$\mathcal{H}_\infty$}-optimal interval observer synthesis for
  uncertain nonlinear dynamical systems via mixed-monotone decompositions.
\newblock {\em IEEE Control Systems Letters}, 6:3008--3013, 2022.

\bibitem{khajenejad2020simultaneousfullrank}
M.~Khajenejad and S.Z. Yong.
\newblock Simultaneous input and state interval observers for nonlinear systems
  with full-rank direct feedthrough.
\newblock In {\em 2020 59th IEEE Conference on Decision and Control (CDC)},
  pages 5443--5448. IEEE, 2020.

\bibitem{khajenejad2021simultaneousdeficient}
M.~Khajenejad and S.Z. Yong.
\newblock Simultaneous input and state interval observers for nonlinear systems
  with rank-deficient direct feedthrough.
\newblock In {\em 2021 European Control Conference (ECC)}, pages 2311--2316.
  IEEE, 2021.

\bibitem{khajenejad2022resilient}
M.~Khajenejad and S.Z Yong.
\newblock Resilient state estimation and attack mitigation in cyber-physical
  systems.
\newblock In {\em Security and Resilience in Cyber-Physical Systems: Detection,
  Estimation and Control}, pages 149--185. Springer, 2022.

\bibitem{khajenejad2022simultaneousijrnc}
M.~Khajenejad and S.Z. Yong.
\newblock Simultaneous state and unknown input set-valued observers for
  quadratically constrained nonlinear dynamical systems.
\newblock {\em International Journal of Robust and Nonlinear Control},
  32(12):6589--6622, 2022.

\bibitem{khajenejadtight23}
M.~Khajenejad and S.Z. Yong.
\newblock Tight remainder-form decomposition functions with applications to
  constrained reachability and guaranteed state estimation.
\newblock {\em IEEE Transactions on Automatic Control}, 68(12):7057--7072,
  2023.

\bibitem{HKK:02}
H.~Khalil.
\newblock {\em Nonlinear Systems}.
\newblock Prentice Hall, 2002.

\bibitem{HK-PG-MZ-PL:17}
H.~Kim, P.~Guo, M.~Zhu, and P.~Liu.
\newblock Attack-resilient estimation of switched nonlinear cyber-physical
  systems.
\newblock In {\em {A}merican {C}ontrol {C}onference}, pages 4328--4333, 2017.

\bibitem{LL-WW-QM-KP-XL-LL-JL:21}
L.~Li, W.~Wang, Q.~Ma, K.~Pan, X.~Liu, L.~Lin, and J.~Li.
\newblock Cyber attack estimation and detection for cyber-physical power
  systems.
\newblock {\em Applied Mathematics and Computation}, 400:126056, 2021.

\bibitem{LL-LM-JG-JZ-YB:21}
L.~Liu, L.~Ma, J.~Guo, J.~Zhang, and Y.~Bo.
\newblock Distributed set-membership filtering for time-varying systems: A
  coding--decoding-based approach.
\newblock {\em Automatica}, 129:109684, 2021.

\bibitem{SL-SM-JC:23-tac}
S.~Liu, S.~Mart{\'\i}nez, and J.~Cort\'es.
\newblock Stabilization of linear cyber-physical systems against attacks via
  switching defense.
\newblock {\em IEEE Transactions on Automatic Control}, 2023.
\newblock To appear.

\bibitem{AYL-GHY:17}
A.~Y. Lu and G.~H. Yang.
\newblock Secure state estimation for cyber-physical systems under sparse
  sensor attacks via a switched {L}uenberger observer.
\newblock {\em Information Sciences}, 417:454--464, 2017.

\bibitem{PL-EJVK-CCDV-QC:16}
P.~Lu, E.~J.~Van Kampen, C.~C.~De Visser, and Q.~Chu.
\newblock Framework for state and unknown input estimation of linear
  time-varying systems.
\newblock {\em Automatica}, 73:145--154, 2016.

\bibitem{YL-LZ-XM:13}
Y.~Lu, L.~Zhang, and X.~Mao.
\newblock Distributed information consensus filters for simultaneous input and
  state estimation.
\newblock {\em Circuits, Systems, and Signal Processing}, 32(2):877--888, 2013.

\bibitem{MM-AV:91}
M.~Milanese and A.~Vicino.
\newblock Optimal estimation theory for dynamic systems with set membership
  uncertainty: an overview.
\newblock {\em Automatica}, 27(6):997--1009, 1991.

\bibitem{EM-FY-QLH-LV:18}
E.~Mousavinejad, F.~Yang, Q.~L. Han, and L.~Vlacic.
\newblock A novel cyber attack detection method in networked control systems.
\newblock {\em IEEE Transactions on Cybernetics}, 48(11):3254--3264, 2018.

\bibitem{YN-YM:15}
Y.~Nakahira and Y.~Mo.
\newblock Dynamic state estimation in the presence of compromised sensory data.
\newblock In {\em {IEEE} Int. Conf. on Decision and Control}, pages 5808--5813,
  2015.

\bibitem{CHN-FJW:98}
C.~H. Nien and F.~J. Wicklin.
\newblock An algorithm for the computation of preimages in noninvertible
  mappings.
\newblock {\em International Journal of Bifurcation and Chaos}, 8(02):415--422,
  1998.

\bibitem{TN-AEM:15}
T.~Nishikawa and A.~E. Motter.
\newblock Comparative analysis of existing models for power-grid
  synchronization.
\newblock {\em New Journal of Physics}, 17(1):015012, 2015.

\bibitem{UW:93}
University of~Washington.
\newblock Power systems test case archive.
\newblock 1993.

\bibitem{MP-PT-IL-GJP:15}
M.~Pajic, P.~Tabuada, I.~Lee, and G.~J. Pappas.
\newblock Attack-resilient state estimation in the presence of noise.
\newblock In {\em {IEEE} Int. Conf. on Decision and Control}, pages 5827--5832,
  2015.

\bibitem{FP-FD-FB:12a}
F.~Pasqualetti, F.~Dorfler, and F.~Bullo.
\newblock Attack detection and identification in cyber-physical systems.
\newblock {\em IEEE Transactions on Automatic Control}, 58(11):2715–2729,
  2013.

\bibitem{TP-MK-SPD-SZY:22}
T.~Pati, M.~Khajenejad, S.~P. Daddala, and S.~Z. Yong.
\newblock ${L}_1$-robust interval observer design for uncertain nonlinear
  dynamical systems.
\newblock {\em IEEE Control Systems Letters}, 6:3475--3480, 2022.

\bibitem{VR-BJG-JR-THS:21}
V.~Renganathan, B.~J. Gravell, J.~Ruths, and T.~H. Summers.
\newblock Anomaly detection under multiplicative noise model uncertainty.
\newblock {\em IEEE Control Systems Letters}, 6:1873--1878, 2021.

\bibitem{KRS-QS-SZY:18}
K.~R. Singh, Q.~Shen, and S.~Z. Yong.
\newblock Mesh-based affine abstraction of nonlinear systems with tighter
  bounds.
\newblock In {\em {IEEE} Int. Conf. on Decision and Control}, pages 3056--3061,
  2018.

\bibitem{EDS:04}
E.~D. Sontag.
\newblock Input to state stability: {B}asic concepts and results.
\newblock In {\em Nonlinear and Optimal Control Theory}, Lecture Notes in
  Mathematics. Springer, 2005.

\bibitem{YS-HS:17}
Y.~Sun and H.~Song.
\newblock {\em Secure and trustworthy transportation cyber-physical systems}.
\newblock Springer, 2017.

\bibitem{XW-HS-FZ-GC:23}
X.~Wang, H.~Su, F.~Zhang, and G.~Chen.
\newblock A robust distributed interval observer for {LTI} systems.
\newblock {\em IEEE Transactions on Automatic Control}, 68(3):1337--1352, 2023.

\bibitem{PW-BC-LS-LY:23}
P.~Weng, B.~Chen, S.~Liu, and L.~Yu.
\newblock Secure nonlinear fusion estimation for cyber--physical systems under
  {FDI} attacks.
\newblock {\em Automatica}, 148:110759, 2023.

\bibitem{CW-ZH-JL-LW:18}
C.~Wu, Z.~Hu, J.~Liu, and L.~Wu.
\newblock Secure estimation for cyber-physical systems via sliding mode.
\newblock {\em IEEE Transactions on Cybernetics}, 48(12):3420--3431, 2018.

\bibitem{LY-OM-NO:19}
L.~Yang, O.~Mickelin, and N.~Ozay.
\newblock On sufficient conditions for mixed monotonicity.
\newblock {\em IEEE Transactions on Automatic Control}, 64(12):5080--5085,
  2019.

\bibitem{LY-NO:19}
L.~Yang and N.~Ozay.
\newblock Tight decomposition functions for mixed monotonicity.
\newblock In {\em {IEEE} Int. Conf. on Decision and Control}, pages 5318--5322,
  2019.

\bibitem{SZY:18}
S.~Z. Yong.
\newblock Simultaneous input and state set-valued observers with applications
  to attack-resilient estimation.
\newblock In {\em {A}merican {C}ontrol {C}onference}, pages 5167--5174, 2018.

\bibitem{SZY-MQF-EF:16}
S.~Z. Yong, M.~Q. Foo, and E.~Frazzoli.
\newblock Robust and resilient estimation for cyber-physical systems under
  adversarial attacks.
\newblock In {\em {A}merican {C}ontrol {C}onference}, pages 308--315, 2016.

\bibitem{SYZ-MZ-EF:16}
S.~Z. Yong, M.~Zhu, and E.~Frazzoli.
\newblock A unified filter for simultaneous input and state estimation of
  linear discrete-time stochastic systems.
\newblock {\em Automatica}, 63:321--329, 2016.

\bibitem{KZ:16}
K.~Zetter.
\newblock Inside the cunning, unprecedented hack of {U}kraine's power grid.
\newblock Wired Magazine, 2016.

\bibitem{JZ-AGE-MN-LM-AA-VT-IK-BP-AKS-JQ-ZH-APM:19}
J.~Zhao, A.~Gomez-Exposito, M.~Netto, L.~Mili, A.~Abur, V.~Terzija, I.~Kamwa,
  B.~Pal, A.~K. Singh, J.~Qi, Z.~Huang, and A.~P. Meliopoulos.
\newblock Power system dynamic state estimation: Motivations, definitions,
  methodologies and future work.
\newblock {\em IEEE Transactions on Power Systems}, 34:3188--3198, 07 2019.

\end{thebibliography}

\appendices
\appendix
\subsection{Matrices \& Parameters}\label{sec:matrices}

\subsubsection{\textbf{Matrices in Lemma \ref{lem:equiv} and its proof in Appendix \ref{sec:lem1_proof}}}
\label{sec:lem1_mat}
\begin{gather*}
  M^i_1 \triangleq (\Xi^i)^{-1}, \quad
  \eta^i_{k+1} \triangleq
                 \begin{bmatrix}
                   (w_k)^\top & (v^i_k)^\top & (v^i_{k+1})^\top
                 \end{bmatrix}^\top, \\
  M^i_2 \triangleq (C^i_2G^i_2)^\dagger, \quad
          \Phi^i \triangleq (I-G^i_2M^i_2C^i_2)G^i_1M^i_1, \\
  \zeta^i_{k+1} \triangleq T^i\Phi^iz^i_{1,k}
                  + L^iz^i_{2,k}
                  + (T^iG^i_2M^i_2 + \Gamma^i)z^i_{2,k+1}, \\
  \Psi^i \triangleq
           \begin{bmatrix}
             T^iB^i & -(T^i\Phi^iD^i_1+L^iD^i_2) & -(T^iG^i_2M^i_2+\Gamma^i)D^i_2
           \end{bmatrix}.
\end{gather*}

\subsubsection{\textbf{Matrices in Equation \eqref{eq:framers}}}
\label{sec:framer_mat}
\begin{align*}
  \tilde{A}^i &\triangleq T^iA^i-L^iC^i_2, \quad
  \mathtt{\tilde{A}}^i
  \triangleq \begin{bmatrix}
                (\tilde{A}^i)^\oplus & -(\tilde{A}^i)^\ominus \\
                -(\tilde{A}^i)^\ominus & (\tilde{A}^i)^\oplus
              \end{bmatrix}, \\
    \mathtt{T}^i
    &\triangleq \begin{bmatrix}
                 (T^i)^\oplus & -(T^i)^\ominus \\
                 -(T^i)^\ominus & (T^i)^\oplus
               \end{bmatrix}, \
  \ol{\eta}^i
  \triangleq \begin{bmatrix}
                \ol{w}^\top & (\ol{v}^i)^\top & (\ol{v}^i)^\top
              \end{bmatrix}^\top, \\
    \ul{\eta}^i
    &\triangleq \begin{bmatrix}
                 \ul{w}^\top & (\ul{v}^i)^\top & (\ul{v}^i)^\top
               \end{bmatrix}^\top, \
  \mathtt{\Psi}^i
  \triangleq \begin{bmatrix}
                (\Psi^i)^\oplus & -(\Psi^i)^\ominus \\
                -(\Psi^i)^\ominus & (\Psi^i)^\oplus
              \end{bmatrix}.
\end{align*}

\subsubsection{\textbf{Matrices in Equation \eqref{eq:d_x}}}
\label{sec:input_mat}
\begin{align*}
  \Upsilon^i &\triangleq (V^i_2M^i_2C^i_2G^i_1-V^i_1)M^i_1, \quad
  \Theta^i \triangleq -V^i_2M^i_2,\\
  \zeta^i_{d,k+1} &\triangleq  \Theta^i z^i_{2,k+1}-\Upsilon^i z^i_{1,k}, \quad
  \Lambda^i \triangleq \begin{bmatrix} C^i_h & \Upsilon^iD^i_1 & \Theta^iD^i_2 \end{bmatrix}.
\end{align*}

\subsubsection{\textbf{Matrices in Equation \eqref{eq:input_framers}}}
\label{sec:input_mat2}
\begin{align*}
  \mathtt{A}^i_h \triangleq
  \begin{bmatrix}
    (A^i_h)^\oplus & -(A^i_h)^\ominus \\
    -(A^i_h)^\ominus & (A^i_h)^\oplus
  \end{bmatrix}, \quad
  \mathtt{\Lambda}^i \triangleq
  \begin{bmatrix}
    (\Lambda^i)^\oplus & -(\Lambda^i)^\ominus \\
    -(\Lambda^i)^\ominus & (\Lambda^i)^\oplus
  \end{bmatrix}.
\end{align*}

\subsubsection{\textbf{Matrices in Lemma \ref{lem:switched}}}
\label{sec:error_mat}
% \margin{Recall that $\tilde{A}^i$ is in the previous list of matrices
%   2), and then, I haven't ever seen any $\ol{F}^i$ up to this point so
%   far. They must be in the proof.  Please define them here as well. }
% Scott: they are defined here now. Need to discuss when/where to redef
\begin{gather*}
  \Acompx^i \triangleq |\tilde{A}^i| + |T^i|\ol{F}_{\rho,x}, \quad
  \Acompd^i \triangleq |{A}^i_h|^ + \ol{F}_{\mu,x}^i, \\
  \Bcompx^i \triangleq |\Psi^i| +
  \begin{bmatrix} |T^i|\ol{F}_{\rho,w} & 0 & 0 \end{bmatrix}, \\
  \Bcompd^i \triangleq |\Lambda^i| +
  \begin{bmatrix} \ol{F}_{\mu,w}^i & 0 & 0 \end{bmatrix}, \\
  \ol{F}_{\rho,x} \triangleq (\ol{J}^f_x)^\oplus - (\ul{J}^f_x)^\ominus, \\
  \ol{F}_{\mu,x}^i \triangleq (\Theta^i C_2^i \ol{J}^f_x)^\oplus - (\Theta^i C_2^i \ul{J}^f_x)^\ominus,
\end{gather*}
where recall that $\ol{J}^f_x, \ul{J}^f_x$ are the Jacobians in Assumption~\ref{assumption:mix-lip}.

% \margin{In addition, in Lemma 4 we see variables $\delta_\eta$ for
% the first time. The next subsection identifies some $\delta$'s but
% not a $\delta_\eta$}
% Scott: resolved in Lemma 4 itself

\subsubsection{\textbf{Matrices in Equation \eqref{eq:error_upper_bounds}}}
\label{sec:bound_mat}
\begin{align*}
  \pi^i_x &\triangleq |\Psi^i|\delta^i_{\eta}+|T^i|\ol{F}_{\rho,w}\delta_w, \\
  \pi^i_d &\triangleq |\Lambda^i|\delta^i_{\eta}+|T^i|\ol{F}_{\mu,w}^i\delta_w, \\
  \delta^i_{\eta} &\triangleq \ol{\eta}^i-\ul{\eta}^i, \quad
  \delta_{w} \triangleq \ol{w}-\ul{w}, \\
  \ol{F}_{\rho,w} &\triangleq (\ol{J}^f_w)^\oplus - (\ul{J}^f_w)^\ominus, \\
  \ol{F}_{\mu,w}^i &\triangleq (\Theta^i C_2^i \ol{J}^f_w)^\oplus - (\Theta^i C_2^i \ul{J}^f_w)^\ominus.
\end{align*}
% \noindent\textbf{\textit{Matrices in Equation \eqref{eq:stab_min_LMI}.}}

% \margin{The first two lines of these section make use of a
%   $\delta_\eta$, however, we define a $\delta_\mu$ here, not
%   $\delta_\eta$. Please fix this}
% Scott: resolved in Lemma 4.

\subsection{Proof of Lemma \ref{lem:equiv}}\label{sec:lem1_proof}
{First,} note that from \eqref{eq:z1k} and with $M^i_1 \triangleq (\Xi^i)^{-1}$, $d^i_{1,k}$ can be computed as a function of the current time state as in \eqref{eq:di1}.
This, in combination with \eqref{eq:stateq} and
\eqref{eq:z2k} results in
\begin{align*}
  M^i_2z^i_{2,k+1} &= M^i_2(C_2^ix_{k+1}+D_2^iv^i_{k+1}) \\
                   &=  M^i_2(C_2^i(f(x_k,w_k)
                     + G^i_1(M^i_1(z^i_{1,k} \\
                   &\quad - C_1^ix_k - D_1^iv^i_{k})
                     + G^i_2d^i_{2,k})
                     + D_2^iv^i_{k+1}),
\end{align*}
where $M^i_2$ is defined in Appendix~\ref{sec:lem1_mat}, which given Assumption~\ref{ass:rank}, returns \eqref{eq:di2}.

By plugging $d_{1,k}^i$ and $d_{2,k}^i$ from \eqref{eq:di1} and
\eqref{eq:di2} into \eqref{eq:stateq}, we have
\begin{align}
  \label{eq:stateqq}
  \begin{array}{rl}
    x_{k+1} &= {f}^i(x_k,w_k)+\Phi^i(z^i_{1,k}-D^i_1v^i_k)\\
            &+G^i_2M^i_2(z^i_{2,k+1}-D^i_2v^i_{k+1}),
\end{array}
\end{align}
where $\Phi^i$ is defined in Appendix~\ref{sec:lem1_mat} and
\begin{gather*}
  {f}^i(x,w) \triangleq f(x,w)-\Phi^i C^i_1x.
\end{gather*}
% \margin{I can't find $C_0^i$, should that be $C_1^i$?}
% Scott: resolved
Combined with the fact that $T^i =I-\Gamma^iC^i_2$, this implies
\begin{align}
  \label{eq:interm_1}
  x_{k+1} = T^i({f}^i(x_k,w_k)+\hat{z}^i_{k+1}+\hat{v}^i_{k+1})
  + \Gamma^iC_2^ix_{k+1},
\end{align}
where
\begin{align*}
  \begin{array}{rl}
\hat{z}^i_{k+1} &\triangleq \Phi^iz^i_{1,k}+G^i_2M^i_2z^i_{2,k+1},\\
    \hat{v}^i_{k+1} &\triangleq -(\Phi^iD^i_1v^i_k+G^i_2M^i_2D^i_2v^i_{k+1}).
  \end{array}
\end{align*}
% \margin{I'm rephrasing this}
% Scott: looks fine
Applying the JSS decomposition described
in Proposition \ref{prop:JSS_decomp} to the vector field $f^i$, there
are matrices $A^i, B^i$ and a remainder vector field $\rho^i(x,w)$,
that allow us to decompose $f^i$ as:
\begin{align*}
  f^i(x,w)=A^ix+B^iw+\rho^i(x,w).
\end{align*}
Now, plugging in $C^i_2x_{k+1}=z^i_{2,k+1}-D^i_2v^i_{k+1}$ from
\eqref{eq:z2k} into \eqref{eq:interm_1},
adding the \emph{zero term}
$ L^i(z^i_{2,k}-{C^i_2x_k -D^i_2v^i_{2,k}})=0$ to both sides of
\eqref{eq:interm_1}, and employing the previous JSS
decomposition in the same expression, returns the results in
\eqref{eq:gamma_dyn_2}. \qed

\subsection{Proof of Lemma \ref{lem:switched}}\label{sec:lem4_proof}
% Recall that $X = X^\oplus - X^\ominus$ for any matrix $X$.
% Scott: not sure why this is here
% Sonia: probably because of the application of proposition 2, (and referring to its proof). But it's not necessary I think.''
Our starting point is equation~\eqref{eq:framers}, and recall the expression of the matrices in Appendix~\ref{sec:error_mat}. First, by Proposition~\ref{prop:tight_decomp},
\begin{gather*}
  \rho^i_d(\ol{x}_k,\ol{w},\ul{x}_k,\ul{w})
  - \rho^i_d(\ul{x}_k,\ul{w},\ol{x}_k,\ol{w})
  \leq \ol{F}_{\rho,x}e_{x,k}
  + \ol{F}_{\rho,w}\delta_w.
\end{gather*}
By subtracting the
top and bottom expressions in \eqref{eq:framers}, and grouping terms in the resulting equation, we conclude that
\begin{align}\label{eq:errorx_1}
  e^0_{x,k+1} \leq \Acompx e_{x,k} + \gamma^x_k,
\end{align}
for some appropriate variables $\gamma^x_k$.
% \margin{maybe we don't have to define every new variable, so I've added the previous phrase}
% Scott: looks good
Further, by the construction of $\sigma^x_k$, applying the $\min$
and $\max$ operations in \eqref{eq:network_update}, the state errors can be
equivalently represented as
\begin{align}\label{eq:errorx_2}
  e_{x,k} = \sigma^x_k e^0_{x,k}.
\end{align}
In a similar manner, subtracting the top and bottom of
\eqref{eq:input_framers}, as well as bounding the nonlinear terms as
above (after replacing $\rho$ with $\mu$), yields
\begin{align}
  \label{eq:errord_1}
  e^0_{d,k} \leq \Acompd e_{x,k} +
  \gamma^d_k,
\end{align}
for some $\gamma^d_k$, while applying the $\min$ and $\max$ operations in
\eqref{eq:network_update_d} returns
\begin{align}\label{eq:errord_2}
  e^d_{x,k} = \sigma^d_k e^0_{d,k}.
\end{align}
Combining \eqref{eq:errorx_1}--\eqref{eq:errord_2} yields
\eqref{eq:error-switched}.\qed

\subsection{Proof of Theorem \ref{thm:stability}}
\label{sec:Thm1_proof}
We first prove sufficiency and then necessity.

As for the sufficiency, assume there is a $\sigma_*^x \in \Sigma^x$ such that $\sigma^x_*\Acompx$ is Schur
stable.  Consider the comparison system
$\tilde e_{x,k+1} = \sigma^x_*\Acompx \tilde e_{x,k}$ with
initial condition $\tilde e_{x,0} = e_{x,0}$.  By the construction of
$\sigma^x_k$ in \eqref{eq:H}, it holds that
$\sigma^x_*\Acompx e_{x,k} \ge
\sigma^x_k\Acompx e_{x,k}$,
$\tilde e_{x,k} \ge e_{x,k} \ge 0$ for all $k \ge 0$ by induction.
Therefore, by the comparison lemma, \eqref{eq:error-switched} is globally
exponentially stable.
% \margin{let's say ``by the comparison lemma''}
% Scott: changed, but is that specific to ODEs?
To prove necessity, assume that~\eqref{eq:error-switched} is asymptotically stable. However, this is the case only if the lower spectral radius of
$\mathcal{F}$ is less than 1. By Proposition
\ref{lem:lower-spectral-radius}, this implies existence of a stable
$F_*= \sigma^x_*\Acompx$.

Finally, having studied stability of the noise-free system, we now
study the C-ISS property of the noisy system in
\eqref{eq:error-switched}.  As before, we can use the comparison
system
\begin{align}\label{eq:iss-comp}
  \tilde{e}_{x,k+1} = \sigma^x_*(\Acompx\tilde{e}_{x,k} +
  \gamma^x_k), \quad \tilde e_{x,0} = e_{x,0}.
\end{align}
It is well known that stable LTI systems are ISS
\cite{EDS:04}. Again, \eqref{eq:H} guarantees
$\tilde e_{x,k} \ge e_{x,k} \ge 0 \ \forall k \ge 0$ by induction,
regardless of the values of the bounded augmented noise
$\gamma^x_k$. By this comparison, the C-ISS property of the system
\eqref{eq:iss-comp} implies that \eqref{eq:error-switched} is
C-ISS.\qed

\subsection{Proof of Lemma~\ref{lem:opt_cent}}
\label{sec:opt-cent}
By \cite[Proposition 1]{XC-JL-PL-ZS:13}, $\rho(\sigma\Acompx) < 1$ if and
only if there exists $p > 0$ such that
$p^\top(\sigma \Acompx -I) < 0$. Using Theorem~\ref{thm:stability},
this implies that \eqref{eq:error-switched} is ISS.

The bound \eqref{eq:error-bound} follows directly from \cite[Theorem
2]{XC-JL-PL-ZS:13}.

\subsection{Proof of Theorem~\ref{thm:milp}}
\label{sec:milp}
First, we introduce a diagonal matrix $Q \in \real^{Nn \times Nn}$ so
that $p = Q \mathbf{1}_{Nn}$. Then we introduce the modified decision
variables $\tilde{L} = QL$, $\tilde{\Gamma} = Q\Gamma$, and
$\tilde{T} = QT$. These give rise to the new matrices
$\tilde{\mathcal{A}} = Q\mathcal{A}$ and
$\tilde{\mathcal{B}} = Q\mathcal{B}$. Next, we can rewrite the
nonlinear terms containing $\sigma \tilde{\mathcal{A}}$ and $\sigma \tilde{\mathcal{B}}$ using
the so-called ``big-$M$'' formulation \cite{JNH:12} to see that
$\sigma\tilde{\mathcal{A}} = \mathbf{A}$ if and only if for all $i \in \nodes$ and
all $j \in \mathcal{N}_i$,
\begin{gather*}
    -(I-\sigma_{ij})M \le \mathbf{A}_{ij} - \Acompx^j \le (I-\sigma_{ij})M \\
    \text{and } -\sigma_{ij} M \le \mathbf{A}_{ij} - \Acompx^j \le \sigma_{ij} M,
\end{gather*}
as long as $M > \max_{i,j} (\tilde{\Acompx})_{ij}$.
In the same way we see that $\sigma\tilde{\mathcal{B}} = \mathbf{B}$ iff for all $i \in \nodes$ and
all $j \in \mathcal{N}_i$,
\begin{gather*}
    -(I-\sigma_{ij})M \le \mathbf{B}_{ij} - \tilde{\mathcal{B}}^j \le (I-\sigma_{ij})M \\
    \text{and } -\sigma_{ij} M \le \mathbf{B}_{ij} - \tilde{\mathcal{B}}^j \le \sigma_{ij} M,
\end{gather*}
as long as $M > \max_{i,j} (\tilde{\mathcal{B}})_{ij}$. Combining all
these transformations and requiring that
$M >\max (\max_{i,j} (\tilde{\mathcal{A}})_{ij},\max_{i,j}
(\tilde{\mathcal{B}})_{ij})$ ensures a one-to-one correspondence
between the original constraints in \eqref{eq:cent-objective} and the
MILP formulation in \eqref{eq:cent-objective_2}-\eqref{eq:mixed_integer}.
\qed

% \margin{This proof may seem obvious, but it is incomplete, please add
%   more detail. First, $M$ is chosen here in a different way than how
%   it is chosen in the statment (it needs to be larger, to treat the
%   bilinear terms of $B$) In addition these inequalities do not use the
%   tilde versions of the $\Acompx$ or $\Bcompx$ (which appear in the
%   statement)}
% \margin{regarding the problem objective: why do we keep
%   a $I_{Nn}$ in the term $\mathbf{A} - I_{Nn}$? shouldn't $Q$ appear
%   in the objective as well?  Given some other previous reviews, please
%   walk the reader slower through the problem transformation.  $Q>0$
%   leads to the matrices $\widetilde{\Acompx}$, $\widetilde{\Bcompx}$,
%   from $\Acompx$ and $\Bcompx$, correct?. Is that relationship
%   obvious? (it may not be for some reviewers)}
% Scott: finished the proof

\subsection{Proof of Theorem \ref{thm:cpdn_verif}}\label{sec:peas_proof}
We will construct $\sigma^x_*$, which by Theorem \ref{thm:stability}
is sufficient for the C-ISS property to hold. For each node
$i\in\nodes$ and state dimension $s \in \{1,\dots,n\}$, using
$\nu_{is}$ from Assumption \ref{ass:col_pos_det_neigh},
\begin{align*}
    (\sigma^x_*)_{\id(i,s),\id(\nu_{is},s)} = 1, \  \
\end{align*}
and all other entries are zero. Since $\nu_{is} \in \mathcal{N}_i$,
$\sigma^x_*$ is a member of $\Sigma_x$ by construction. With
$\sigma^x_*$ defined as such, row $\id(i,s)$ of $\sigma^x_*\Acompx$ is
equal to row $\id (\nu_{is},s)$ of $\Acompx$
(cf. Lemma~\ref{lem:switched}). From the definition of
$\Acompx^i = |\tilde{A}^i| + |T^i|\ol{F}_{\rho,x}$
it is clear that
$\|(\tilde A^i)_s\|_1+\| (|T^i| \ol{F}_{\rho,x})_s\|_1$ =
$\|(\mathcal{A}^i_x)_s\|_1$.

% \margin{add a line in the notations about the meaning of the sub $s$
% and $s+n$ (row of a matrix)}
% Scott: done
Note that the gains $T^i$ and $L^i$ are computed by
\eqref{eq:stab_LMI}, which independently minimize the sum of the
$1$-norm of each row of $\tilde A^{i}$ and the $1$-norm of the same
row of $|T^i| \ol{F}_{\rho,x}$, since the ${s}^\text{th}$ rows
of $T^i$ and $L^i$ only affect the ${s}^\text{th}$ row of
$\tilde A^i \triangleq T^iA-L^iC_2^i$, as well as the ${s}^\text{th}$
row of $|T^i| \ol{F}_{\rho,x}$. Moreover, Assumption
\ref{ass:col_pos_det_neigh} guarantees
$\|(\Acompx^{\nu_{is}})_s\|_1 < 1$ for each $s$. All of this
implies
$ \|(\sigma^x_*\Acompx)_{\id (i,s)}\|_1 < 1$.  Since this
holds for every row of the matrix $\sigma^x_*\mathcal{A}^x$, then
$\rho(\sigma^x_*\Acompx) \leq
\|\sigma^x_*\Acompx\|_\infty \triangleq \max\limits_{1\leq
  i \leq nN} \sum_{s=1}^{nN} |(\sigma^x_*\Acompx)_{ij}|<
1$.\qed

\subsection{Proof of Lemma \ref{lem:errors}}
\label{sec:error_bounds_proof}
Starting from the error dynamics \eqref{eq:error-switched}, and given
the stability of the observer that is guaranteed by
\eqref{eq:stab_LMI} (cf. Theorem \ref{thm:cpdn_verif}), for any
$\sigma^x\in \Sigma^x,\sigma^d\in \Sigma^d$, the
framer error dynamics can be bounded as follows:
\begin{align*}
  e_{x,k+1}\leq \sigma^x(\Acompx e_{x,k} +
  \gamma^x_k), \
  e_{d,k} \leq  \sigma^d(\Acompd e_{x,k} + \gamma^d_{k}).
\end{align*}
Then, it follows from the solution of the above system that:
\begin{align}\label{eq:comparison_bounding}
    e_{x,k} \leq (\sigma^x\Acompx)^{k-1} e_{x,0}
    + \sum_{j=1}^{k-1}(\sigma^x\Acompx)^{k-j}\gamma^x_{j-1}.
\end{align}
Further, leveraging the noise bounds, we obtain:
\begin{align}\label{noise_x_bound}
  \begin{array}{rl}
    \|\gamma^x_k\|_{\infty} &\leq
                              \max_{i}\||\Psi^i|\delta^i_{\eta}+|T^i|\ol{F}_{\rho,w}\delta_w\|_{\infty},\\
    \|\gamma^d_k\|_{\infty} &\leq
                              \max_{i}\||\Lambda^i|\delta^i_{\eta}+\ol{F}_{\mu,w}^i\delta_w\|_{\infty},
  \end{array}
\end{align}
where
\begin{align*}
  \delta^i_{\eta} \triangleq \ol{\eta}^i -
  \ul{\eta}^i, \  \delta_{w} \triangleq \ol{w} -
  \ul{w}.
\end{align*}
The results follow from \eqref{eq:comparison_bounding},
\eqref{noise_x_bound}, sub-multiplicativity of norms and the triangle
inequality.\qed

\subsection{Proof of Theorem \ref{thm:suff_stability}}
\label{sec:error_min_proof}
Assumption \ref{ass:col_pos_det_neigh} implies the existence of gains
that render the {\dsiso} algorithm C-ISS.  It remains to show that the
solutions of \eqref{eq:stab_min_LMI} are stabilizing.  First, notice
that Algorithm \ref{alg:init} computes $\mathbb{J}^i$ {by solving
  \eqref{eq:stab_LMI}.}  The use of $\mathbb{J}^i$ in the constraints
of \eqref{eq:stab_min_LMI} guarantees that the optimization problem is
feasible. Furthermore, we can show that since
Assumption~\ref{ass:col_pos_det_neigh} holds, there exists
$\sigma^x_*$ such that $\rho(\sigma^x_*\Acompx) < 1$, and therefore
that the {\dsiso} algorithm is ISS. We refer the reader to Theorem
\ref{thm:cpdn_verif} for the details of the construction of
$\sigma^x_*$. This in combination with Lemma \ref{lem:errors} ensures
that the bounds in \eqref{eq:error_upper_bounds} converge to their
steady state values in \eqref{eq:ss_upperbounds}.\qed

 {\tiny
\begin{IEEEbiography}
 % [{\includegraphics[width=1in,height=1.3in,keepaspectratio]{pic_moh_3.jpg}}]
  [{\includegraphics[width=1in,height=1.3in,keepaspectratio]{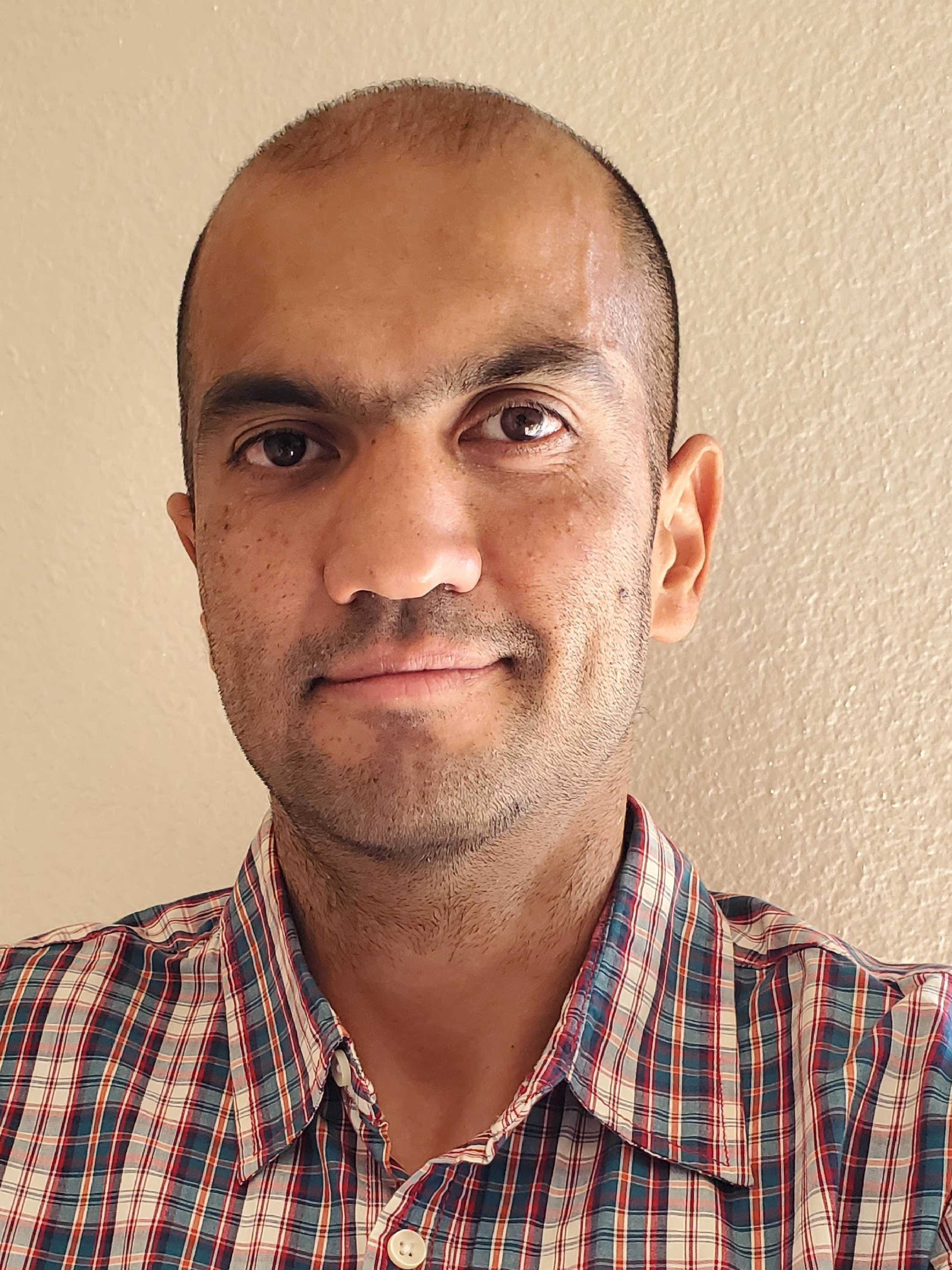}}]
  {Mohammad Khajenejad}
  is a postdoctoral scholar in the Mechanical and Aerospace
  Engineering Department at University of California, San Diego, CA,
  USA. He received his Ph.D. in Mechanical Engineering from Arizona
  State University, Tempe, AZ, USA, in 2021, where he won the ASU
  Dean's Dissertation Award for his Ph.D. dissertation. Mohammad
  received his M.S. and B.S. in Electrical Engineering from The
  University of Tehran, Iran. He is the author or co-author of diverse
  papers published in refereed conference proceedings and
  journals. His current research interests include set-theoretic
  control, resiliency and privacy of networked cyber-physical systems
  and robust game theory.
\end{IEEEbiography}
\vskip -2\baselineskip plus -1fil
\begin{IEEEbiography}
 % [{\includegraphics[width=1in,height=1.3in,keepaspectratio]{pic_scott.jpg}}]
  [{\includegraphics[width=1in,height=1.3in,keepaspectratio]{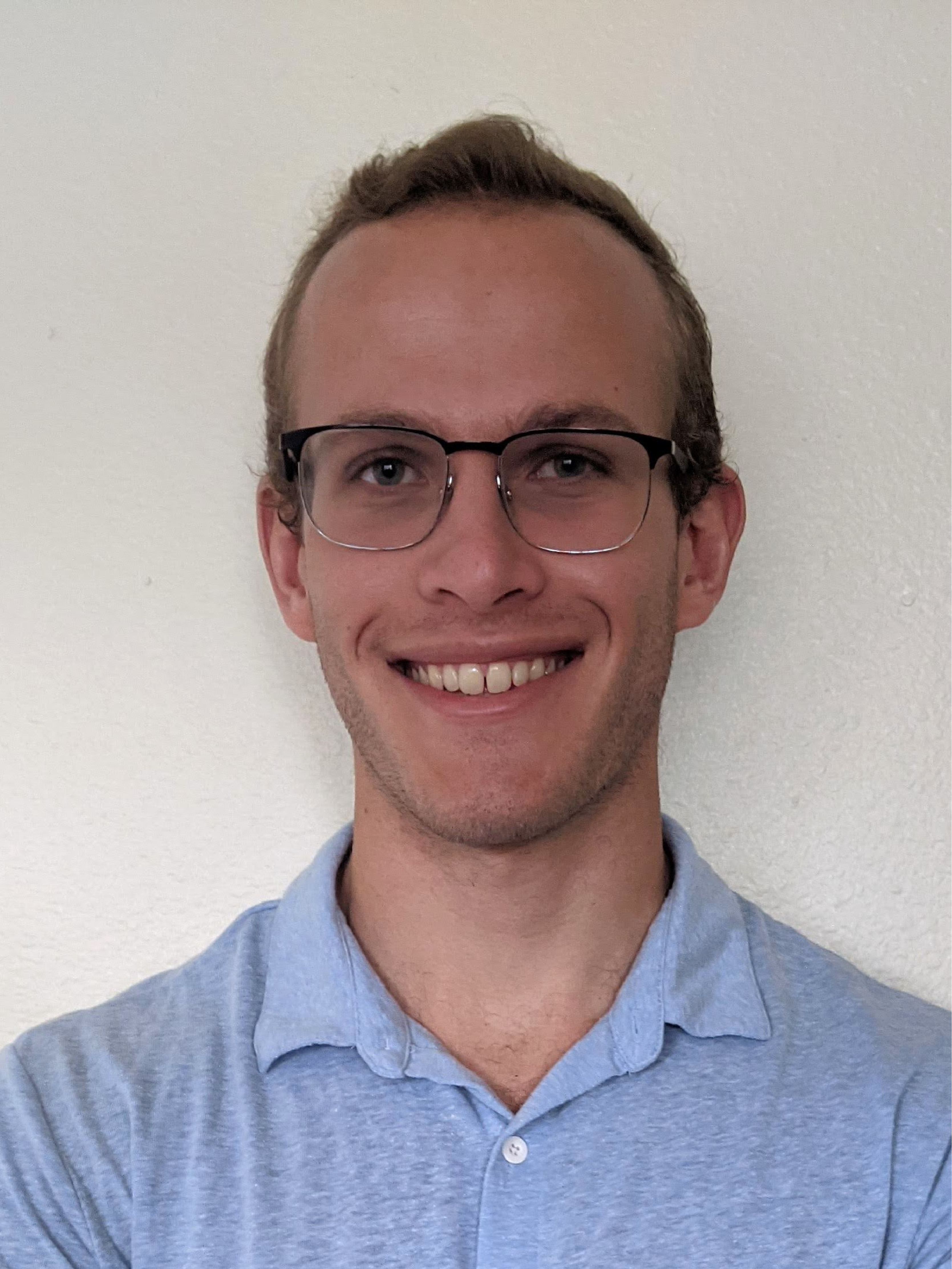}}]
  {Scott Brown}
  is a Ph.D. student in the Mechanical and Aerospace Engineering
  Department at University of California, San Diego, CA, USA, advised
  by Sonia Mart{\`\i}nez. He received his B.S. in Aerospace
  Engineering from the University of California, San Diego. His
  research interests include control and state estimation in networked
  systems, robust control using set-theoretic methods, nonlinear
  control, and optimization. He is a student member of IEEE and the
  IEEE Technical Committee on Hybrid Systems.
\end{IEEEbiography}
\vskip -2\baselineskip plus -1fil
\begin{IEEEbiography}
  %[{\includegraphics[width=1in,height=1.3in,keepaspectratio]{photo-SM.jpg}}]
  [{\includegraphics[width=1in,height=1.3in,keepaspectratio]{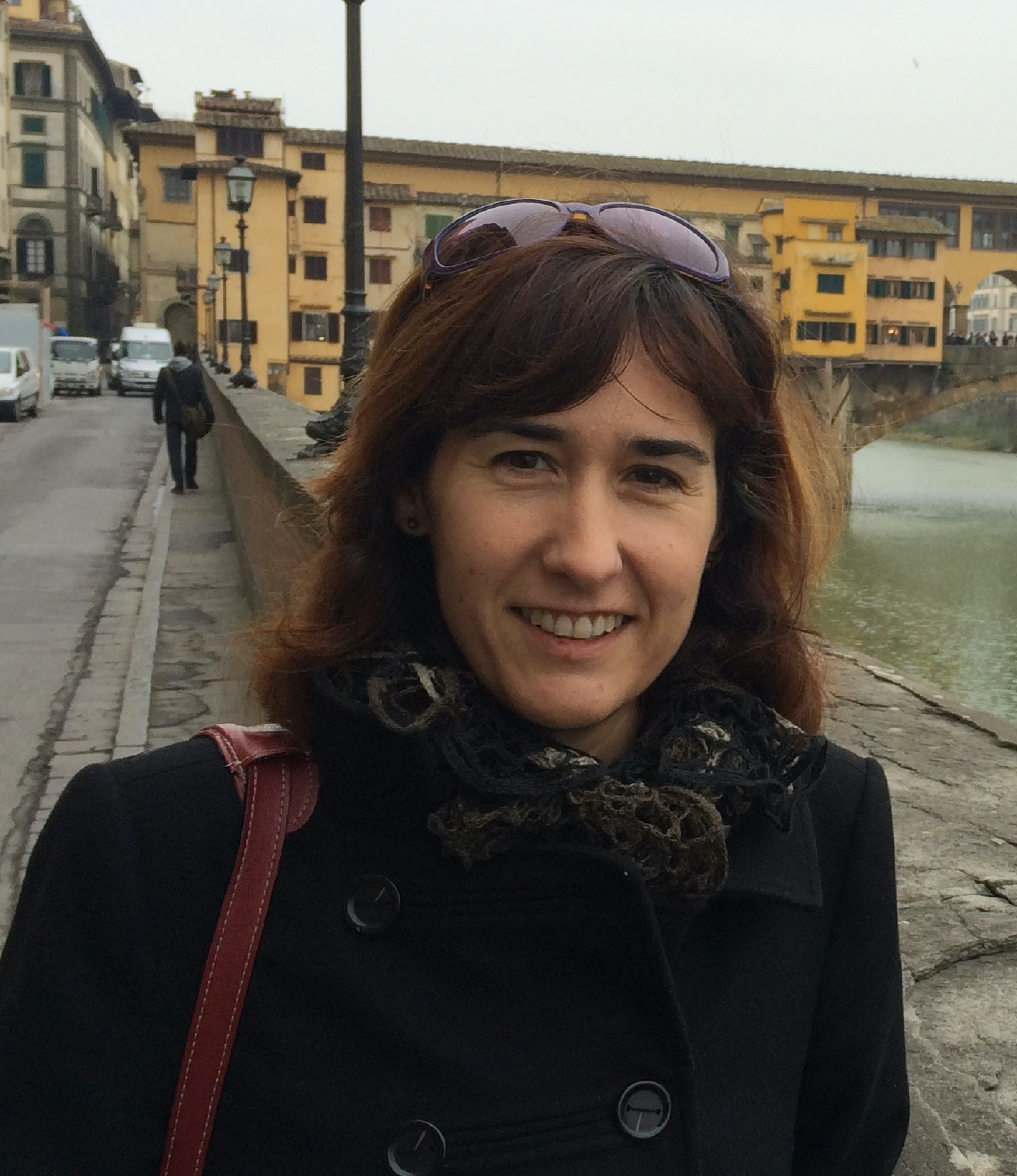}}]
  {Sonia Mart{\'\i}nez} (M'02-SM'07-F'18) is a Professor of Mechanical
and Aerospace Engineering at the University of California, San Diego,
CA, USA. She received her Ph.D. degree in Engineering Mathematics from
the Universidad Carlos III de Madrid, Spain, in May 2002. She was a
Visiting Assistant Professor of Applied Mathematics at the Technical
University of Catalonia, Spain (2002-2003), a Postdoctoral Fulbright
Fellow at the Coordinated Science Laboratory of the University of
Illinois, Urbana-Champaign (2003-2004) and the Center for Control,
Dynamical systems and Computation of the University of California,
Santa Barbara (2004-2005).  Her research interests include the control
of networked systems, multi-agent systems, nonlinear control theory,
and planning algorithms in robotics. She is a Fellow of IEEE. She is a
co-author (together with F. Bullo and J. Cort\'es) of ``Distributed
Control of Robotic Networks'' (Princeton University Press, 2009). She
is a co-author (together with M. Zhu) of ``Distributed
Optimization-based Control of Multi-agent Networks in Complex
Environments'' (Springer, 2015).  She is the Editor in Chief of the
recently launched \textit{CSS IEEE Open Journal of Control Systems.}
\end{IEEEbiography}}
\end{document}